\documentclass{article}
\usepackage{graphicx} 

\usepackage[a4paper,top=2cm,bottom=2cm,left=3cm,right=3cm,marginparwidth=1.75cm]{geometry}

\usepackage{amsmath}
\usepackage{amssymb}
\usepackage{tikz-cd} 
\usetikzlibrary{fadings}
\usepackage{fancyvrb} 
\usepackage{amsthm}
\usepackage{array}

\usepackage[colorlinks=true, allcolors=blue]{hyperref}

\usepackage{cite}

\VerbatimFootnotes  

\newcommand\numberthis{\addtocounter{equation}{1}\tag{\theequation}}

\newcommand{\opar}{\overline{\partial}}
\newcommand{\proj}[1]{\mathbb{CP}^{#1}}
\renewcommand{\d}{\mathrm{d}}
\newcommand{\del}{\partial}
\newcommand{\delbar}{\Bar{\del}}
\newcommand{\C}{\mathbb{C}}
\newcommand{\CP}{\mathbb{CP}}

\newcommand{\cL}{\mathcal{L}}

\newcommand{\vCY}{\text{Vol}_{\text{CY}}}

\newcommand{\gCY}{g_{\text{CY}}}

\newcommand{\hrs}{h_{\rm RootSc}}
\newcommand{\hcan}{h_{\rm hom}}
\newcommand{\hperm}{h_{\rm perm}}
\newcommand{\piperm}{\pi_{\rm perm}}

\newcommand{\cymetric}{\texttt{cymetric} }

\newtheorem{proposition}{Proposition}
\theoremstyle{remark}

\theoremstyle{definition}
\newtheorem{definition}{Definition}[section]

\definecolor{airforceblue}{rgb}{0.36, 0.54, 0.66}

\newcommand{\mo}[1]{{\color{orange}  #1}}

\vspace{1cm}

\title{       {\Large \bf Learning Group Invariant Calabi--Yau Metrics by Fundamental Domain Projections}}

\vspace{2cm}

\author{
Yacoub Hendi\footnote{yacoub.hendi@math.uu.se}, \, Magdalena Larfors\footnote{magdalena.larfors@physics.uu.se}, \, Moritz Walden\footnote{ moritz.walden@physics.uu.se}
}

\date{\today}

\begin{document}

\maketitle
\begin{center} {{\it Department of Physics and Astronomy \& Department of Mathematics, 
\\Uppsala University,\\
       L\"agerhyddsv. 1, SE-751 20 Uppsala, Sweden.}}\\

\end{center}

\begin{abstract}
\noindent
We present new invariant machine learning models that approximate 
the Ricci-flat metric on Calabi--Yau (CY)  manifolds with discrete symmetries. We accomplish this by combining the $\phi$-model of the \cymetric package with non-trainable, $G$-invariant, canonicalization layers that project the $\phi$-model's input data (i.e. points sampled from the CY geometry) to the fundamental domain of a given symmetry group $G$. These $G$-invariant layers are easy to concatenate, provided one compatibility condition is fulfilled, and combine well with spectral $\phi$-models. Through experiments  on different CY geometries, we find that, for fixed point sample size and training time, canonicalized models give  slightly more accurate metric approximations than the standard $\phi$-model.  The method may also be used to the compute Ricci-flat metric on smooth CY quotients. We demonstrate this aspect by experiments on a smooth $\mathbb{Z}^2_5$ quotient of a 5-parameter quintic CY manifold.

\end{abstract}

\thispagestyle{empty}
\setcounter{page}{0}
\newpage

\tableofcontents
\clearpage

\clearpage
\section{Introduction}
\label{sec:intro}

String compactifications lie at the heart of many formal and phenomenological studies of string theory. By compactifying some 
of the spatial dimensions of the space-time probed by string theory one can construct models that match observations from particle physics and cosmology, or discover how different limits of string theory are related by duality. The topology and geometry formed by the compactified dimensions are key to these constructions, and there is an interesting interplay between physics and geometric constraints. Among the compact spaces, Calabi--Yau manifolds provide a sweet spot; their mathematics is highly non-trivial, yet examples can be constructed in the thousands.

Calabi--Yau  (CY) manifolds are complex, K\"ahler manifolds, with  vanishing first Chern class; by the Calabi--Yau theorem these spaces admit a Ricci-flat metric, which is unique in a given K\"ahler class.  Yau \cite{Yau:1978cfy} has proven that the K\"ahler form $J$ associated to the Ricci-flat CY metric $\gCY$ satisfies a complex Monge-Amp\`ere (MA) equation, which it reads, in $n$ complex dimensions,
\begin{equation}
        J \wedge J \wedge ...\wedge J = \kappa \, \Omega \wedge \overline{\Omega} \; ,
\end{equation}
where we wedge $n$ copies of $J$ on the left hand side, $\Omega$ denotes the unique holomorphic top form on the CY manifold, and $\kappa$ is a complex constant. There is no known analytic solution to this equation for $n > 1$.\footnote{In one complex dimension, the CY is a torus, which admits a flat metric. In 2D, i.e. for a K3 surface, a duality argument has been proposed to determine the Ricci-flat metric \cite{Kachru:2018van,Kachru:2020tat}.}  

Lacking an analytic solution to a partial differential equation (PDE), one may try to find numerical solutions. This is indeed possible, owing to the fact that, since $J$ is related to any reference K\"ahler for $J'$ via
\begin{equation}
        J  = J' + i \partial \bar{\partial} \phi \; ,
\end{equation} 
the MA equation can be rewritten as a second order PDE for a real function $\phi$. 
One approach to solve the PDE for $\phi$, is to expand $\phi$ in  degree-$k$ polynomials on the CY, an idea that goes back to Tian \cite{tian90}, and then solve for the expansion coefficients. This can be done using a fixed-point iteration scheme, known as the Donaldson algorithm, \cite{donaldson2005numerical,Douglas:2006hz,Douglas:2006rr,Braun:2007sn}, or via functional minimization \cite{Headrick:2005ch,Headrick:2009jz,Cui:2019uhy}; collectively,  the metrics obtained in this way are called {\it algebraic} metrics. Since the polynomial, or {\it spectral}, basis is complete in the limit $k \to \infty$ \cite{donaldson2005numerical,Headrick:2005ch}, this method will indeed give the unique Ricci-flat metric in this limit. However, for most CY spaces, the spectral basis dimension grows rapidly with $k$, implying that practical applications are limited to rather small $k$.\footnote{For a CY hypersurface in $\mathbb{CP}^N$, the dimension of the polynomial basis scales like $k^{2N}$ for large $k$ \cite{Headrick:2009jz}.} For symmetric CY geometries, using an invariant spectral basis, of smaller dimension, allows us to  compute the metric at larger $k$  (in fact, the most accurate approximations of CY metrics are obtained on highly symmetric Fermat quintics \cite{Headrick:2005ch}).\footnote{However, constructing the invariant basis sometimes comes at a non-trivial computational cost, as we will discuss in section \ref{sec:freecanon}.}

With the advent of machine learning (ML), novel opportunities, based on neural networks, have arisen to approximate CY metrics. Since neural networks are universal  approximators \cite{Cybenko:1989aaa,Leshno:1993aaa,Pinkus:1999aaa}, they can learn any function; thus $\phi$, or why not the Ricci-flat CY metric directly, may be learned by an ML model. Developing machine learning methods for CY geometry has been an active field of research in later years \cite{Ashmore:2019wzb,Anderson:2020hux,Jejjala:2020wcc,Douglas:2020hpv,Douglas:2021zdn,Ashmore:2021ohf,Larfors:2021pbb,Larfors:2022nep,Gerdes:2022nzr,Berglund:2022gvm,Ahmed:2023cnw,Halverson:2023ndu,Anderson:2023viv}. A major practical gain of ML methods is that they are less dependent on expansions in the spectral basis. While ML methods may be used in combination with the spectral expansion of $\phi$, they can also make away completely with this crutch.  As such, the ML methods remain computationally feasible on generic CY manifolds, as evidenced in Refs.~\cite{Larfors:2021pbb,Larfors:2022nep,Anderson:2023viv}. 

Admittedly, the accuracy of ML algorithms has been observed to improve on more symmetric spaces \cite{Douglas:2020hpv,Douglas:2021zdn,Berglund:2022gvm,Anderson:2023viv}, and the training time (for similar performance) is slightly longer on non-symmetric spaces \cite{Anderson:2023viv}. These differences, which are small compared to the effects seen in pre-ML methods, indicate that the ML models learn the symmetries of the data they are trained on. Symmetries have also been used to test the accuracy of ML predictions of CY metrics \cite{Jejjala:2020wcc}. However, discrete symmetries have, so far, not played any significant role in the design  of ML models for CY metrics. 
In this paper, we will perform a first exploration of this topic, using novel methods to construct group invariant ML models.\footnote{In ML methods based on spectral basis, such as \cite{Anderson:2020hux,Gerdes:2022nzr} and \cite{Douglas:2020hpv,Douglas:2021zdn}, one can restrict to group invariant elements of the basis; to the knowledge of the authors this has not been explored to any greater degree.}  

The motivation for our study is two-fold. First, as we will review below, in the ML literature there is ample evidence that for symmetric data, equivariant and invariant models outperform  networks that are agnostic about symmetries in the data.  The Geometric Deep Learning program \cite{bronstein7974879,bronstein2021geometric} shows how symmetries explain the success of state-of-the-art ML models such as  CNNs, GNNs, and Transformers; ref.~\cite{Cheng:2019xrt} analyse the CNNs using the concept of covariance in theoretical physics.  Formally, it has been proven that group invariant neural networks are universal approximators for group invariant functions \cite{yarotsky:2022}, thus symmetry-enforcing ML models are as expressive as the symmetry-agnostic ones.

Secondly, from the perspective of string theory, there is  phenomenological motivation to study non-simply connected CY geometries, obtained by quotienting a CY manifold by a freely acting discrete symmetry $G$. $G$-invariant ML models provide a new way to compute the Ricci-flat metrics of such quotients. This problem has been explored using Donaldson's algorithm in the past \cite{Braun:2007sn}, with positive effects on the accuracy of the metric approximation. Ref.~\cite{Braun:2007sn} also explains the technical subtleties associated with the creation of a group invariant spectral basis; this provides further motivation to explore invariant ML methods, as do recent application of ML methods to the  computations of heterotic Yukawa couplings on CY manifolds with discrete symmetries \cite{Butbaia:2024tje,Constantin:2024yxh}.

\paragraph{Present work}
In this work, we will design group invariant ML models that provide approximations of the Ricci-flat metric of CY manifolds. Our work is motivated by the importance of symmetries in the context of spectral methods for CY metrics, and the role symmetries play in leading ML models. 

For our studies, we will explore CY geometries created as complete intersections in complex projective spaces, so-called CICYs \cite{Candelas:1987kf}. We will use the \cymetric package, introduced in \cite{Larfors:2021pbb,Larfors:2022nep}, which allows us to sample points on such spaces, and provide several ML model architectures that either predict the function $\phi$, or the metric directly. In this work, we will limit our attention to the $\phi$-model. To create group invariant $\phi$-models, we will include non-trainable layers into the neural network architecture. These layers pre-process the
the input data of the $\phi$-model, and ensure that it is $G$-invariant by projecting to a fundamental domain of the group $G$. This idea, which is often referred to as {\it canonicalization}, has recently been proposed in \cite{aslan2022group,kaba2023equivariance}. Our main contribution is a comparison of the performance and accuracy of the canonicalized ML models with the standard $\phi$-model. We also discuss the similarities and differences of the canonicalization approach with the spectral networks suggested by \cite{Berglund:2022gvm}, and how canonicalization can be combined with spectral $\phi$-models. Finally, we study how canonicalized ML models can be used to compute the metric on a non-simply connected CY manifold.\footnote{The code for the canonicalization layers can be found in \url{https://github.com/jake997/invariant_layers_cymetric}.}

This paper is organised as follows. In section \ref{sec:relwork}, we briefly review related work on machine learned, Ricci-flat CY metrics, and outline the key approaches to group invariant ML models. Section \ref{sec:mlcy} describes CY geometry and the ML model we will use for our experiments. We also recall how the spectral layer is constructed. In section \ref{sec:fundproj}, we present our main contribution: how to construct canonicalization layers for the symmetry groups relevant for CY point samples; experiments using these techniques are presented and discussed in section \ref{sec:exp}. Some technical aspects of our analysis are described in appendices. In section \ref{sec:freecanon}, we apply canonicalization layers to the computation of Ricci-flat CY metrics on smooth quotients of a 5-parameter quintic CY. Conclusions and discussions are given in section \ref{sec:conclude}.

\section{Related work}
\label{sec:relwork}

The past five years have witnessed a flurry of research activities on ML methods for Ricci-flat CY metrics \cite{Ashmore:2019wzb,Anderson:2020hux,Jejjala:2020wcc,Douglas:2020hpv,Douglas:2021zdn,Ashmore:2021ohf,Larfors:2021pbb,Larfors:2022nep,Gerdes:2022nzr,Berglund:2022gvm,Ahmed:2023cnw}, and there are several recent examples in the literature on how these ML methods may be used to address open questions in string theory \cite{Ahmed:2023cnw,Butbaia:2024tje,Constantin:2024yxh,Ashmore:2021qdf,Douglas:2024pmn}. Since the Ricci-flat CY metric is unknown, all ML models for  CY metrics rely on some sort of semi-supervised learning, where a custom loss function is  used to ensure that the learned metric is Ricci-flat. This is in reminiscent of Physics Informed Neural Networks (PINNs) \cite{RAISSI2019686}; in the parlance of PINNs, the Ricci-flatness  is encoded as a soft constraint. 

To date,  three open-source ML libraries have been developed. While we will only explore one of these libraries, the key features of the other packages may also be of interest to the reader, and are therefor given here 
\begin{itemize}
\item The \texttt{MLgeometry} package was developed for refs.~\cite{Douglas:2020hpv,Douglas:2021zdn,Douglas:2024pmn}. This package introduces the bihomogeneous (holomorphic) neural networks where the layers are bihomogeneous (holomorphic) functions of the ambient space coordinates for the CY manifold, with quadratic non-linear activation functions. These layers encode the spectral basis of the algebraic CY metrics, and the learned metric is  automatically  globally defined and K\"ahler (i.e. these are hard constraints). The Ricci-flatness condition is imposed softly via a loss function. 
\item The \cymetric package,  presented in refs.~\cite{Larfors:2021pbb,Larfors:2022nep} and further tested in ref.~\cite{Schneider:2022ssn}, instead encodes the CY metric, or alternatively the K\"ahler function, directly in a dense, shallow, fully connected neural network (typically with ReLU or GELU activation functions). Custom loss functions are used to enforce that the learned metric is globally defined, Kähler, Ricci-flat, and restricted to the K\"ahler class selected when starting the training.\footnote{The motivation for the soft encoding of all mathematical constraints on the metric, was to make the package generalize, at little cost, to other geometric constructions, beyond the CY regime.}
\item the \texttt{cyjax} package \cite{Gerdes:2022nzr} is a spectral method. The ML model learns  algebraic CY metrics; the neural net encodes an Hermitian coefficient matrix, which, when combined with the spectral basis, provides an approximation to the Ricci-flat metric. This automatically provides a globally defined, K\"ahler metric; the Ricci-flatness condition is imposed via a  loss function. 
\end{itemize}
In ref.~\cite{Berglund:2022gvm}, it was proposed that the $\phi$-model of the \cymetric package can be augmented by  a non-linear, non-trainable {\it spectral layer}, where the input data is passed to the $k=1$ spectral basis, and then fed to the dense, trainable layers of the model. The benefit of this layer is that the predicted metric is automatically globally defined and Kähler. Thus, some of the soft constraints of the $\phi$-model become hard constraints with the addition of a spectral layer. In ref.~\cite{Berglund:2022gvm}, this was shown to be of essence in learning the Ricci-flat metric on singular, or nearly singular, CY manifolds.  The spectral layer is, from a machine learning perspective, a data pre-processing step; as we will discuss below it shares some features with the ideas explored in this paper. 

In this paper, we will be interested in predicting a function $\phi: X \to \mathbb{R}$, where $X$ is a CY manifold. Working locally, we may decompose $X$ into patches which are isomorphic to $\mathbb{R}^n$. The prediction of a function
\[
f: \mathbb{R}^n \, \to \, \mathbb{R}
\]
is a standard task in machine learning. When the function domain has symmetries, invariant (and equivariant) machine learning models are of interest. This topic is extensively studied. An excellent starting point for studies of the mathematical foundations of invariant and equivariant ML models can be found in the book \cite{bronstein2021geometric}, which also provides a comprehensive list of references to this field of data science.

There are different approaches to the construction of invariant machine learning models. Following Yarotsky \cite{yarotsky:2022}, see also \cite{aslan2022group}, one may distinguish between {\it symmetrization-based} and {\it intrinsic} approaches. In the first approach, one may, for example, construct $G$-invariant models by augmenting the training data by $G$-transformed copies with the same label as the original data point.  The network will then learn that data points related by $G$-transformations give the same output, as required by invariance. When applied to unseen data, a network trained on augmented data will perform better than a network trained on non-augmented data. As discussed in ref.~\cite{JMLR:v21:20-163}, from a mathematical point of view, and under exact invariance, data augmentation is equivalent to averaging over orbits of the group $G$; it can thus be proven that this reduces the variance of the prediction of a model.

Intrinsic approaches instead produce invariant networks by imposing restrictions on the network design.  CNNs, which are equivariant under translations, are early incarnations of this idea \cite{FUKUSHIMA1982455,Lecun98}. Other network designs that are invariant under group transformations include \cite{cohen2016arXiv160207576C,cohen2018arXiv180110130C,Gens14,Ravanbakhsh2017arXiv170208389R,maron:2020}; these and related ideas are further discussed in \cite{bronstein2021geometric} and the review \cite{Gerken:2021sla}. Common to these designs is that in- or equivariance is  implemented by weight-sharing in some layer of the network.

Restricting the trainable weights of some layer in the network is not the only intrinsic way to achieve invariant/equivariant layers. Yarotsky has shown, using classical polynomial invariant theory, that an extra polynomial layer which is invariant/equivariant and complete, results in an invariant/equivariant network \cite{yarotsky:2022}. This is particularly useful for finite groups, where complete invariant layers can be constructed from a finite number of polynomial invariants. For continuous groups, invariant polynomial layers can also be constructed, but, lacking completeness, one cannot prove that they are universal invariant approximators \cite{yarotsky:2022}. The spectral layer \cite{Berglund:2022gvm} (described briefly above and in detail below) is, up to rescaling by the norm, an incarnation of such an invariant layer for the continuous group $\mathbb{C}^*$.  

Of key interest to us is the alternative intrinsic approach proposed in \cite{aslan2022group,kaba2023equivariance}. Here, a pre-processing step is used to  put the data  into a canonical form, by projection to a fundamental domain of the symmetry group.  The canonicalization  can be achieved in different ways, and we will describe this method  in more detail below. Importantly, the method   works  for both continuous and discrete groups. It is also good to note that, in contrast to weight-sharing schemes, that methods based on polynomial invariants or canonicalization are not restricted to models using neural networks. Since the concatenation of an invariant map with any other map preserves invariance, such data pre-processing steps can be used to impose invariance on virtually any type of ML model, for example the packages developed for machine learning of CY metrics.

\section{Machine learning Ricci-flat CY metrics}
\label{sec:mlcy}

In this section, we review some aspects of the construction of examples of CY manifolds, and their symmetries. A more detailed discussion of the symmetries is relegated to appendix \ref{ap:isometry}. We then describe the key points of the \cymetric package, and the ML methods we employ to learn CY metrics. For more details on the \cymetric package, we refer to \cite{Larfors:2021pbb,Larfors:2022nep}.

\subsection{CY construction}

For our experiments, we need explicit realizations of CY geometries. 
In this paper, we will restrict our studies to the Complete Intersection Calabi-Yau (CICY) manifolds \cite{Candelas:1987kf,Green:1987cr}.\footnote{It would be interesting to extend our studies to CY manifolds in toric ambient spaces. The \cymetric package has point generators also for such spaces, and the methods we develop here can be adapted to this setting. We leave this generalisation for future work.} These are constructed as the zero loci of a set of homogeneous polynomials $p_i$, $i=0,\dots,N$, in an ambient space $\mathbb{CP}^{n_1} \times ...\times \mathbb{CP}^{n_N}$, such that $\dim(CY)=n_1+\dots+n_N-N$. A prime example is the Fermat quintic, defined by the solutions of
\begin{equation} \label{eq:fermquin}
   \sum_{i=0}^4 z_i^5 = 0
\end{equation}
in $\mathbb{CP}^4$. CICY manifolds are often described in terms of configuration matrices, which in the case of the quintic is $[4|5]$. 
This class of CY manifolds is fully classified, and a detailed list of their topological properties can be found in \cite{lukas-webpage}. 

CICY manifolds may, at certain subloci in their complex structure moduli space, admit discrete symmetries.\footnote{We define CY $3$-folds to have strict $SU(3)$ holonomy; consequently, the spaces we study do not admit continuous isometries.} These are manifest as invariances of the defining polynomials of the CICY. For example, the defining equation of the Fermat quintic \eqref{eq:fermquin}
is invariant under permutations of the coordinates, multiplications of the coordinates by 5th roots of unity (up to a global multiplication by the same root, which is removed by the scaling invariance of the homogeneous coordinates), and complex conjugation; the full symmetry group  is $\mathbb{Z}_5^4 \rtimes (S^5 \times \mathbb{Z}_2)$ \cite{Braun:2007sn}. 
The freely acting symmetries of CICY manifolds have been classified in \cite{Braun:2010vc}, non-freely acting symmetries are classified in \cite{Lukas:2017vqp}, and symmetries of CICY quotients are classified in \cite{Candelas:2017ive}. The invariance of the defining polynomials translates into an isometry of the Ricci-flat CY metric (see appendix \ref{ap:isometry} for a proof).  

We want to learn the Ricci-flat metric on some CICY manifold $X$ with dimension $n$.  Starting from a reference K\"ahler form $J'$, we know, by the CY theorem, that there exists a unique K\"ahler form $J$, in the same cohomology class as $J'$, which gives a Ricci-flat metric. 

For simplicity, let us discuss the case where $X$ is realized as one hypersurface in one projective space. All our statements generalize to the case of complete intersections of hypersurfaces, so this is just for ease of presentation. Let $\iota :X\hookrightarrow\mathbb{CP}^{n+1}$ be the inclusion of $X$ in its ambient space, and $J'$ be the pullback of the K\"ahler form $\omega$ associated to the Fubini-Study metric on $\mathbb{CP}^{n+1}$ to $X$, i.e. $J'=\iota^{\star} \omega$.  Then we want to learn a real function $\phi:X\to \mathbb{R}$ such that
\begin{equation}\label{eq:phimodel}
    J = J' + i \, \partial \opar \phi
\end{equation}
is Ricci-flat. We use the \verb|cymetric| library\footnote{\url{https://github.com/pythoncymetric/cymetric}} \cite{Larfors:2021pbb,Larfors:2022nep} for learning said metric, including the $\phi$-\textit{model}, which encodes eq. (\ref{eq:phimodel}).

\subsection{The $\phi$-model of the cymetric package}\label{sec:phimodel}

The \cymetric package consists of point generators, that create samples of points on the CY manifold, and a number of machine learning models that take the coordinates of the sampled points as input, and gives some data related to the Ricci-flat CY metric as output. In the $\phi$-model, the output is a numerical estimate of the real function $\phi$, defined in eq.~\eqref{eq:phimodel}.

Since data pre-processing is a theme of this paper, we need to carefully describe the input data of the $\phi$-model. As mentioned, this data represents points on the CY manifold, as constructed by the point  generators of the \cymetric package; these routines are discussed in depth in ref.~\cite{Larfors:2022nep}.  For simplicity, we continue to focus on the case where $X$ is a hypersurface in $\CP^{n+1}$. Points in the CY $n$-fold $X$ are represented by homogeneous coordinates, $z$, of the ambient space 
$\mathbb{CP}^{n+1}$,  and are repackaged, by the point generators,  as a real tuple
\[
\left(\Re(z_0) \quad \dots \quad \Re(z_{n+1}) \quad \Im(z_0) \quad \dots \quad \Im(z_{n+1}) \right) \in \mathbb{R}^{2(n+2)} \, ,
\]
where the homogeneous scaling invariance of $\mathbb{CP}^{n+1}$ has been used to  ensure that the input data take values in the range $[0,1]$. In accordance with the philosophy of this paper, we will still refer to the network as a map $f:X\to \mathbb{R}$.

In order to train the neural network, we rely on loss functions encoding the relevant mathematical constraints. In the \cymetric package, these constraints are encoded in five custom loss functions. First, we want to ensure Ricci-flatness. As mentioned in the introduction, this can be phrased as requiring $J$ to obey the {Monge--Amp\`ere equation}, which in 3 dimensions reads
\begin{equation}
    J \wedge J \wedge J = \kappa \, \Omega \wedge \overline{\Omega} \,,
\end{equation}
where $\kappa$ is some complex constant and $\Omega$ is the no-where vanishing holomorphic $(3,0)$-from.\footnote{Here $J$ is predicted by the network and $\kappa$ and $\Omega$ are computed by the point generator.} With this at hand, the first loss function is defined as
\begin{equation}
    \mathcal{L}_{\textrm{MA}} = \left\lVert 1 - \frac{1}{\kappa} \frac{J \wedge J \wedge J}{\Omega \wedge \overline{\Omega}} \right\rVert_1 \,,
\end{equation}
where $\lVert \cdot \rVert_1$ denotes the $L_1$ norm.\footnote{The \cymetric package can use any $L_n$ norm for the loss functions, we use the default setting of the package.}

Furthermore, we want to ensure that $\phi$ is a global function. Hence we define the transition loss
\begin{equation}
    \mathcal{L}_{\textrm{transition}} = \frac{1}{d}  \sum_{(s,t)} \left\lVert \phi^{(t)} - \phi^{(s)} \right\rVert_1  \,,
\end{equation}
where $s,t$ denote different patches and $d$ denotes the number of patch transitions.

In the \cymetric package, three more constraints are encoded as custom loss functions $\mathcal{L}_i$: the K\"ahler constraint $\d J = 0$ gives $\cL_{\text{dJ}}$, preserving the K\"ahler class is enforced by $\cL_{\text{Kclass}}$\footnote{The  K\"ahler class loss function is necessary when learning the metric on any CY manifold with $h^{(1,1)}>1$. It reduces to a volume-preserving constraint when  $h^{(1,1)}=1$. The  \cymetric routines require such a loss since they can be applied to any CICY manifold, and CY manifolds defined as hypersurfaces in the toric ambient spaces contained in the Kreuzer--Skarke list \cite{Kreuzer:2000xy}.} and the vanishing Ricci-scalar is given by $\cL_{\text{Ricci}}$. The total loss function is thus given by
	\begin{align}
		\label{eq:loss}
		\cL &= \alpha_1 \cL_{\text{MA}} + \alpha_2 \cL_{\text{dJ}} + \alpha_3 \cL_{\text{transition}} + \alpha_4 \cL_{\text{Ricci}} + \alpha_5 \cL_{\text{Kclass}} \; .
	\end{align}
 \mo{}
Here, $\alpha_i$ are tuneable hyperparameters of the model that are set by the user. In our experiments, we use the default values for the \cymetric losses in the $\phi$-model; thus the K\"ahler and Ricci losses are disabled and $\alpha_i=1$.  We refer the reader to \cite{Larfors:2021pbb,Larfors:2022nep} for more discussions of the loss functions.

\subsection{Spectral layers}
\label{sec:spectral}

The spectral network, introduced in \cite{Berglund:2022gvm}, repackages the input data of the $\phi$-model of the \cymetric package in a form that is manifestly invariant under homogeneous rescalings of the ambient space coordinates.  The motivation for this feature engineering step is to ensure that the learned function $\phi$ is globally well-defined. Recall that this is only enforced "softly", by a loss function, in the \cymetric package. As a consequence, the network may predict line bundle sections rather than globally defined functions. An observed ill consequence of this is the discrepancy of topological quantities computed with the predicted Ricci-flat metric \cite{Berglund:2022gvm}, which becomes particularly noticeable for singular, or nearly singular, CY spaces. 

It may nevertheless appear surprising that such feature engineering is needed. The homogeneous rescaling is not a symmetry of the CY manifold; this is a complex, $n$-dimensional space with $n$ holomorphic coordinates. However,  as we have described in the previous section, the input data of the \cymetric models is not an $n$-dimensional complex tuple. 
Rather, for a CICY with ambient space $\CP^{n_1}\times..\times \CP^{n_k}$, the point generators of  \cymetric produces $\sum_i 2(n_i+2) $-dimensional real input data tuples  (with entries in the range $[0,1]$).  In this pre-processing the points on the CY space are separated into coordinate patches (that overlap but for zero measure sets).  The training is then done locally, using the transition loss function to impose that the function matches on patch overlaps. It is an empirical fact that the transition loss function is around $10^{-3}$ for standard $\phi$-models (trained for 10 epochs on the Fermat quintic). As we will see in section \ref{sec:exp}, with the addition of a spectral layer in the network, the transition loss function is of the order $10^{-9}$ already at initialization of the network, which indicates the prediction of $\phi$ is indeed a function.

The effect of the spectral layer is thus to impose invariance under homogeneous rescalings. Following ref.~\cite{Berglund:2022gvm} (see also \cite{Ashmore:2021qdf}), we have encoded the layer as follows.\footnote{This encoding shows the relation between the spectral layers and the  bihomogeneous layers of \cite{Douglas:2020hpv}; they are essentially related by a division by the norm. However, the bihomogeneous layers have trainable weights, whereas the spectral layer is non-trainable.}  For a CY with ambient space $\mathbb{CP}^n$, the spectral layer is a map 
\[
(z_0 , \hdots, z_n) \mapsto \begin{pmatrix}
    \frac{z_0 \bar{z}_0}{|z|^2} &  \frac{z_0 \bar{z}_1}{|z|^2} & \hdots &\frac{z_0 \bar{z}_n}{|z|^2}\\
    \frac{z_1 \bar{z}_0}{|z|^2} &  \frac{z_1 \bar{z}_1}{|z|^2} & \hdots &\frac{z_1 \bar{z}_n}{|z|^2}\\
    \vdots & \vdots & \ddots & \vdots\\
    \frac{z_n \bar{z}_0}{|z|^2} &  \frac{z_n \bar{z}_1}{|z|^2} & \hdots &\frac{z_n \bar{z}_n}{|z|^2}\\
\end{pmatrix}
\]
where $|z|^2 = \sum_{i=0}^n z_i \bar{z}_i$. This is clearly invariant under 
\[
z \to \lambda z, \lambda \in \C^\star \; .
\]
The entries in the spectral layer matrix corresponds to the degree-1 basis of eigenfunctions of the Laplacian on $\CP^4$ introduced in \cite{Ashmore:2021qdf}; this is the rationale for the name of the layer. Comparing to the machine learning literature, this is almost, up to the division by the norm, an instance of the polynomial invariant layers studied by \cite{yarotsky:2022}.

For CICY manifolds with product ambient spaces $\CP^{n_1} \times \CP^{n_2} \times ... \CP^{n_k}$ the spectral layer is implemented in a block diagonal fashion,\footnote{Note, that these matrices are never explicitly constructed in our implementation. Instead, the input is directly computed as a flattened vector, where the zero entries are discarded, that can be fed into the neural network.} with separate maps for the homogeneous coordinates associated to each $\CP^{n_j}$. 

As it is evident from the formulae above, our spectral layer uses complex inputs. To match the real input of the \cymetric package, we use a conversion between the real inputs  to complex and then back to real.
Whenever we refer to the spectral layer below, we refer to this composition. 

Note that there are no trainable parameters in spectral layers. Their only effect is to map the input data to a convenient, canonical form, that is invariant under homogeneous rescalings. Since the composition of an invariant map with any other map, results in a map with the same invariance, the inclusion of a spectral layer results in an invariant machine learning model, as illustrated in figure \ref{fig:invNN}.

In the next section, we will explore other types of invariant layers, that serve the purpose of making the network prediction $\phi$ invariant under isometries of the CY manifold.

\begin{figure}
    \centering
    \includegraphics[width=0.7\textwidth]{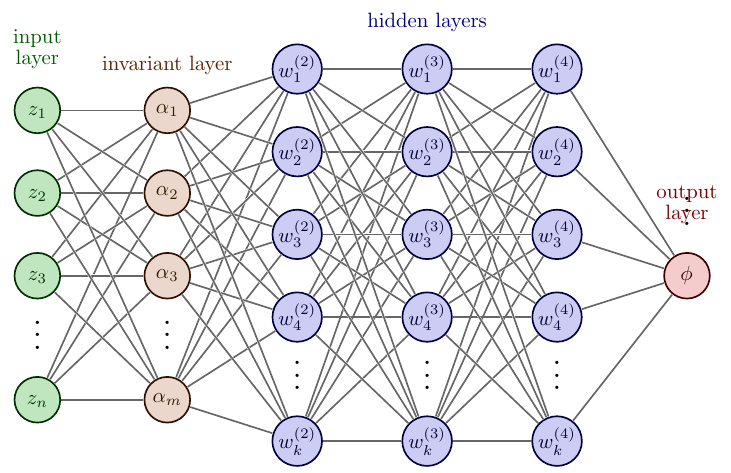}
    \caption{A dense neural network with an invariant layer. The invariant layer does not have trainable parameters. Its sole purpose is to put the input data, $z_i$, in a  form, $\alpha_i(z)$ that is invariant under the relevant symmetry group.}
    \label{fig:invNN}
\end{figure}

\section{Invariant metric by fundamental domain projections}
\label{sec:fundproj}

In the following, we will discuss a method to construct a group invariant neural network predicting the Ricci-flat metric on some CY manifold. This will be done by projecting the input data, i.e. the tuple containing the homogeneous coordinates of the sampled points on the CY space, onto the fundamental domain of said group, following the strategy of \cite{aslan2022group,kaba2023equivariance}.   
We recall that the rationale of these methods is based on the fact that the composition of a group invariant map with any other map, results in a map with the same invariance. Thus, by including a non-trainable, $G$-canonicalization layer, as introduced below, before the dense layers of the $\phi$-model, we get a group invariant network. An illustration of this model architecture is given in figure \ref{fig:invNN}.

We consider a neural network $f:X \to \mathbb{R}$ for some CY manifold $X$. Given some function $h:X\to X$, satisfying
\begin{equation}
    h\left( g \cdot x \right) = h(x) \,, \qquad \forall x \in X \,, \; \forall g \in G \, , 
\end{equation}
for a group $G$ acting on $X$, one can construct a $G$-invariant neural network by
\begin{equation}
    f\left( h\left( x \right) \right) = y \,, \qquad x \in X, \, y \in \mathbb{R}\, .
\end{equation}

 This can be achieved by defining a map $\pi:X\to G$, obeying 
 \begin{equation} \label{eq:equivcondition}
\pi(g\cdot x) = \pi(x)g^{-1} \,, \qquad \forall x \in X \,, \; \forall g \in G\,,
 \end{equation}
 such that
\begin{equation}
    h(x) =  \pi(x) \cdot x \,.
\end{equation}
We call the map $h$ a \emph{ $G$-canonicalization} and  it is easy to see that $h$ is a $G$-invariant map
\begin{equation}
    h(g\cdot x) = \pi(g\cdot x) \cdot (g\cdot x) = \pi(x) (g^{-1} g) \cdot x = \pi(x) \cdot x = h(x) \,.
\end{equation}
The remaining task is to define an appropriate map that obeys eq.~(\ref{eq:equivcondition}). We will do so by defining maps to a fundamental domain of $G$.

\begin{definition}[Fundamental domain]
Let $G$ be a group acting on a topological $G$-set $X$. Then a closed subset $\mathcal{F}$ of $X$ is called a fundamental domain of $G$ in $X$ if $X$ is the union of all conjugates of $\mathcal{F}$, i.e.,
\begin{equation}
    X = \bigcup_{g\in G} g \, \mathcal{F} \,,
\end{equation}
and the intersection of any two conjugates has no interior.
\end{definition}

With this definition at hand, we define $h:X\to \mathcal{F}$, such that $h(g\cdot x)=h(x)$ which defines $\pi$ with respect to a choice of $\mathcal{F}$.

It is important to note that the projection into a fundamental domain can give rise to discontinuities. For a given fundamental domain $\mathcal{F}$ of some group $G$ acting on some manifold $X$, there exist points on $X$ which will get mapped to the boundary of $\mathcal{F}$. Moving infinitesimally within $X$ can then result in discontinuous jumps in $\mathcal{F}$. For the metric approximations we construct, this can lead to problems computing global properties, like the volume which require integration over the entire manifold. We will talk about this further in section \ref{sec:exp} and the conclusion.

\subsection{Examples}
In the following examples, we will consider symmetries on the Fermat quintic, whose defining equation was given in eq.~\eqref{eq:fermquin}. Let $f\circ h: X \to \mathbb{R}$ 
be a neural network learning the function $\phi(z)$ in the Ricci-flat metric $g_{\textrm{CY}}(z)= g_{\textrm{FS}}(z)+ \partial\overline{\partial}\phi(z)$. Let $\mathbf{z} \in \mathbb{C}^5$ encode $z$ as a vector and let $\rho:G \to \textrm{GL}(5,\mathbb{C})$ be a representation of some symmetry group $G$. A $G$-invariant network is then constructed by defining $h(\mathbf{z}) = \rho( \pi(z) )\mathbf{z}$ together with an appropriate choice of $\pi(z)$. In the following, we will construct examples considering several different symmetry groups $G$.
In each example, we specify explicitly the action of $G$ on the ambient space, the fundamental domain $\mathcal{F}$ and the projection map $h$ on $\mathcal{F}$.

\subsubsection{Homogeneous rescalings}
\label{sec:hom}
Consider the set of homogeneous rescalings of the homogeneous coordinates of the embedding space of the quintic, $\mathbb{CP}^4$:
\[
[z_0: z_1 : z_2 : z_3 : z_4] \sim \lambda [z_0: z_1 : z_2 : z_3 : z_4] \mbox{, for } \lambda \in \C^\star \; .
\]
The group $G=\mathbb{C}^{\star}$ is a symmetry 
of the input data 
\[
\left(\Re(z_0) \quad \dots \quad \Re(z_4) \quad \Im(z_0) \quad \dots \quad \Im(z_4) \right)\in \mathbb{R}^{10}
\]
of the $\phi$-model. 

We define the fundamental domain
\begin{equation}\label{eq:homCanon}
    \mathcal{F} = \Big\{[z_0:z_1:z_2:z_3:z_4] \in \mathbb{CP}^4 \, \Big| \,\forall i : \, |z_i| \leq 1, \exists i : z_i = 1   \Big\} \, , 
\end{equation}
and the projection map $\hcan$ to be 
\[\numberthis \label{eq:homcan_map}
\hcan: \; [z_0:z_1:z_2:z_3:z_4] \mapsto \left[\frac{z_0}{z_i}: \frac{z_1}{z_i} : \frac{z_2}{z_i} : \frac{z_3}{z_i} : \frac{z_4}{z_i}\right], \forall j\neq i:\ |z_i| \geq |z_j| \, .
\]
In case more than one coordinate has the biggest norm, we choose to divide by the coordinate with smallest index.\footnote{This is straightforwardly implemented in \verb|TensorFlow| by finding the largest coordinate with \linebreak \verb|tf.math.argmax(tf.abs(complex_input))| and then dividing the input by said coordinate.}

Eq.~(\ref{eq:homCanon}) thus uniquely determines the canonicalization maps $h$ and $\pi$ defined in the beginning of this section. 
The action transforms any input into the canonical form of $\max(|z_0|, \dots , |z_4|)=1$. Note that the \verb|cymetric| input of the point generator is already of this canonical form, hence the layer has no effect. However, upon training, points are re-scaled within the transition loss to 
{compare the prediction of $\phi$ in}
different patches. This re-scaling is inverted by $\hcan(\mathbf{z})$, leading to an almost vanishing transition loss. 

As an example of the action of this map, consider the point 
\[
z = [z_0: z_1 : z_2 : z_3 : z_4]  \mbox{, where } \forall j\neq 3: |z_3| > |z_j| \,  .
\]
Then $\pi(z)=1/z_3$ and in terms $\mathbf{z} \in \mathbb{C}^5$ we have
\begin{equation}
    \hcan(\mathbf{z}) = \rho (\pi ( z ) ) \cdot \mathbf{z} = \begin{pmatrix}
        1/z_3 & 0 & 0 & 0 & 0 \\
        0 & 1/z_3 & 0 & 0 & 0 \\
        0 & 0 & 1/z_3 & 0 & 0 \\
        0 & 0 & 0 & 1/z_3 & 0 \\
        0 & 0 & 0 & 0 & 1/z_3 
    \end{pmatrix}  \begin{pmatrix}
        z_0 \\
        z_1 \\
        z_2 \\
        z_3 \\
        z_4 
    \end{pmatrix} = \begin{pmatrix}
        z_0/z_3 \\
        z_1/z_3 \\
        z_2/z_3 \\
        1 \\
        z_4/z_3 
    \end{pmatrix} \,.
\end{equation}

When the ambient space is a product of complex projective spaces, we apply homogeneous rescaling as described above for the coordinates of each factor, separately.

\subsubsection{Permutations}
\label{sec:perm}
Let $G=S_5$, which acts by permutation of variables. Clearly, the action of this group allows to reorder the entries of a tuple $(z_0, ..., z_n)$ in any way we desire, and we may select a fundamental domain by any consistent order. Thus, given a set of points $[z_0: z_1 : z_2 : z_3 : z_4]$, we define the fundamental domain
\begin{equation} \label{eq:perm_map}
    \mathcal{F} = \Big\{ [z_0: z_1 : z_2 : z_3 : z_4] \in \mathbb{CP}^4 \, \Big| \, |z_0| \geq  |z_1| \geq |z_2| \geq |z_3| \geq |z_4| \Big\} \, .
\end{equation}
The projection map $\hperm$ thus performs the required sorting.\footnote{Again, this is straightforwardly implemented in \verb|TensorFlow| by obtaining the absolute values of the input with \verb|norms=tf.math.abs(complex_inputs)| and then sorting by 
\verb|tf.argsort(norms, axis=-1, direction='DESCENDING')|.
} By default, if the norms are equal, $\hperm$ keeps the original order.

As an example, consider the point $z= [z_0:z_1:z_2:z_3:z_4]$ where $|z_1|>|z_0|>|z_2|>|z_3|>|z_4|$, then
\begin{equation}
    \pi(z) = \begin{pmatrix}
        0 & 1 & 2 & 3 & 4 \\
        1 & 0 & 2 & 3 & 4 
    \end{pmatrix} \,,
\end{equation}
where we have used \textit{Cauchy's two line notation} for $\pi(z) \in S_5$, where the first line corresponds to indices $i$ and the second line to their permutations $\sigma(i)$. In terms of $\mathbf{z} \in \mathbb{C}^5$, we have
\begin{equation}
    \hperm(\mathbf{z}) = \rho (\pi ( z ) ) \cdot \mathbf{z} = \begin{pmatrix}
        0 & 1 & 0 & 0 & 0 \\
        1 & 0 & 0 & 0 & 0 \\
        0 & 0 & 1 & 0 & 0 \\
        0 & 0 & 0 & 1 & 0 \\
        0 & 0 & 0 & 0 & 1 
    \end{pmatrix}  \begin{pmatrix}
        z_0 \\
        z_1 \\
        z_2 \\
        z_3 \\
        z_4 
    \end{pmatrix} = \begin{pmatrix}
        z_1 \\
        z_0 \\
        z_2 \\
        z_3 \\
        z_4 
    \end{pmatrix} \,.
\end{equation}

In the case of product ambient spaces, there may be multiple permutation invariances. Below, we will explore CY manifolds whose ambient space is $(\CP^n)^m$,  
with defining polynomials that are invariant under an
$S_n$ that acts identically on each factor. Concretely, the fundamental domain of such a  diagonal $S_n$ acting on $\CP^n \times \CP^n$   is defined to be 
\begin{equation}\label{eq:FDtwoambient}
    \mathcal{F} = \Big\{ ([z_0: z_1 :\hdots: z_n],[w_0: w_1 :\hdots: w_n] ) \in \mathbb{CP}^n\times \CP^n \, \Big| \, |z_0| \geq  |z_1| \geq \hdots \geq |z_n| \Big\} \, 
\end{equation}
where the restriction is only enforced on the coordinates of the first component.
To project on $\mathcal{F}$, we order the coordinates in the first factor by norms and then use the same ordering to order the rest coordinates of the remaining factors. 
For example, the point $([z_0: z_1 :\hdots: z_n],[w_0: w_1 :\hdots: w_n] ) \in \mathbb{CP}^n\times \CP^n$, such that $|z_1|\geq |z_0| \geq |z_2| \geq |z_3| \geq \hdots \geq|z_n|$ gets mapped under the projection to $([z_1: z_0 :z_2: z_3: \hdots: z_n],[w_1: w_0: w_2 :w_3 :\hdots: w_n] )$. While we do not foresee any large complications, we leave the adaptation of the fundamental domain projections to  ambient spaces of the type $\CP^{m} \times \CP^{n}$ with $m\neq n$, and similar generalizations, to future work.  

\subsubsection{Scaling by roots of unity}
\label{sec:root}

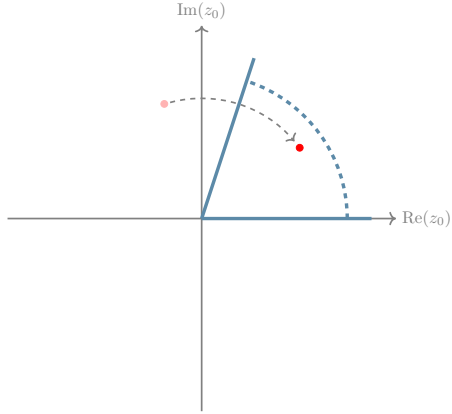
\begin{figure}[!t]
  \centering
  \resizebox{0.4\textwidth}{!}{
  \begin{tikzpicture}
  
  \draw[->, gray, line width=1pt] (3.9,0) -- (4,0) node[anchor=west]{Re($z_0$)};
  \draw[->, gray, line width=1pt] (0,3.9) -- (0,4) node[anchor=south]{Im($z_0$)};
  \draw[gray, line width=1pt] (-4,0) -- (4,0);
  \draw[gray, line width=1pt] (0,-4) -- (0,4);

  \draw[airforceblue, line width=2pt] (0,0) -- (1.08, 3.33);
  \draw[airforceblue, line width=2pt] (0,0) -- (3.5, 0);
  \draw[dashed, airforceblue, line width=2pt] (3,0) arc (0:72:3);
    \filldraw[red, opacity=0.3] (-0.77,2.38) circle (2pt);
    \filldraw[red] (2.02, 1.47) circle (2pt);
    \draw[->, dashed, gray, line width=1pt] (-0.65,2.41) arc(105:40:2.5);

  \end{tikzpicture}}
  \caption{Under scaling by a fifth root of unity, $z_0$ maps into the fundamental domain $\mathcal{F}$ (first quintant) by an appropriate rotation.\label{fig:rootscaling}}
  \end{figure}

Let $G = (\mathbb{Z}_5)^5$ act
 by re-scaling each coordinate by a fifth root of unity, e.g. $z_i \to \exp(2\pi i k/5)z_i$ for some $k\in \{0, \dots, 4\}$. Note, that different coordinates can be re-scaled by different roots simultaneously. We define the fundamental domain 
{\[
\mathcal{F} := \{[z_0:\hdots:z_4] \in \mathbb{CP}\ |\ \forall i: 0\leq arg(z_i)\leq 2\pi/5\} \, ,
\]
and the projection map $\hrs: X \to \mathcal{F}$ is given by\footnote{We implement this projection in \verb|TensorFlow| by finding the angles with \linebreak \verb|angles=tf.math.angle(complex_input)| and then applying the modulo operation with \\
\verb|tf.math.floormod(angles, 2*np.pi/5.)|. }
\[\numberthis \label{eq:rootscalingcan_map}
\hrs : \; [r_0 e^{i\theta_0}: \hdots: r_4 e^{i\theta_4}] \to [r_0\exp[i(\theta_0 \mod 2\pi/5)], \hdots, r_4\exp[i(\theta_4 \mod 2\pi/5)]] \, .
\]
}

As an example, consider the point $z = [r_0 e^{2\pi i \varphi_0}:r_1 e^{2\pi i \varphi_1}:r_2 e^{2\pi i \varphi_2}:r_3 e^{2\pi i \varphi_3}:r_4 e^{2\pi i \varphi_4}]$ where $\varphi_i < 1/5 \; \forall i\neq0$ and $\varphi_0=3/10$. We have for  $\mathbf{z}\in \mathbb{C}^5$ 
\begin{equation}
    \hrs(\mathbf{z}) = \rho (\pi ( z ) )  \mathbf{z} = \begin{pmatrix}
        \exp(2\pi i 4/5) & 0 & 0 & 0 & 0 \\
        0 & 1 & 0 & 0 & 0 \\
        0 & 0 & 1 & 0 & 0 \\
        0 & 0 & 0 & 1 & 0 \\
        0 & 0 & 0 & 0 & 1 
    \end{pmatrix}  \begin{pmatrix}
        r_0 e^{2\pi i 3/10} \\
        r_1 e^{2\pi i \varphi_1} \\
        r_2 e^{2\pi i \varphi_2} \\
        r_3 e^{2\pi i \varphi_3} \\
        r_4 e^{2\pi i \varphi_4} 
    \end{pmatrix} = \begin{pmatrix}
        r_0 e^{2\pi i 1/10} \\
        r_1 e^{2\pi i \varphi_1} \\
        r_2 e^{2\pi i \varphi_2} \\
        r_3 e^{2\pi i \varphi_3} \\
        r_4 e^{2\pi i \varphi_4} 
    \end{pmatrix} \,.
\end{equation}
We illustrate the projection in figure~\ref{fig:rootscaling}.

\subsection{Combining invariant layers}

To achieve metric predictions that are invariant under a (semi)-direct product of groups, as is the case for the Fermat quintic, we would like to construct ML models that are invariant under several groups. This can be attained  by adding the $G$-canonicalization layers in sequence to the network.  However,   we must then make sure that the canonicalizations satisfy a non-trivial condition:
Let $G = G_1 \ltimes G_2$ act on a space $X$ and let $h_i$ be a $G_i$-canonicalization for $i=1, 2$. Then $h:= h_2 \circ h_1$ defines a $G$-canonicalization if and only if $h_1$ descends to a well defined map on the orbits of $G_2$, namely $h_1 : X/G_2 \to X/G_2$ is well defined.
Note that this condition  on $h_1$ is weaker than $h_1$ being equivariant with respect to $G_2$.

To make this concrete, let's return to the metric on the Fermat quintic $X$ defined by \eqref{eq:fermquin}, which is invariant under action of the group $G = S_5 \ltimes \mathbb{Z}_5^5$. For later reference, we note that the action of this group is not free. Let $\hperm$  be the $S_5$-canonicalization defined by \eqref{eq:perm_map}, and let $\hrs$ be a $\mathbb{Z}_5^5$-canonicalization as defined in \eqref{eq:rootscalingcan_map}.
Then $\hrs \circ \hperm$ is a $G$-canonicalization, since  for any $r \in \mathbb{Z}_5^5$ and $z \in X$ we have
\[
\begin{split}
    \hperm(r\cdot z) &= \piperm(r\cdot z) \cdot (r \cdot z)
    =\piperm(z) \cdot (r \cdot z)\\
    &=  \piperm(z) r \piperm(z)^{-1} \piperm(z) \cdot z = \underbrace{\piperm(z) r \piperm(z)^{-1}}_{\in \mathbb{Z}_5^5} \cdot \hperm (z).
\end{split}
\]
Notice that $\piperm(r \cdot z) = \piperm(z) $ because $\piperm$ orders the entries in the $z$-tuple according to their norms and $r$ acting on $z$ only changes their arguments.

Performing a similar analysis, we find that $\hrs \circ \hperm \circ \hcan$ is a $[(\C^\star \times S_5) \ltimes \mathbb{Z}_5^5]$-canonicalization, where $\hcan$ is given by \eqref{eq:homcan_map}, and $\text{Spectral} \circ \hperm$ is invariant under $S_5 \times \C^\star$.

\section{Experiments}
\label{sec:exp}
We have performed a number of experiments, on different CICY geometries, to compare the accuracy and performance of invariant $\phi$-models, obtained through the canonicalization procedure described in section \ref{sec:fundproj}, with the standard and spectral $\phi$-models. In this section, we present the results of these experiments. We use the pointgenerator classes of the \cymetric package to generate samples of 100~000 points on the CY spaces.\footnote{For CY hypersurfaces, we use \texttt{PointGenerator}, for CICYs with product ambient space we use \texttt{CICYPointGenerator}.} The data sets are split (1:9) into validation and training data sets. 
The experiments consist of runs over five different point sets, and we report the average and the standard deviation of these runs. 
We test the models on an unseen test set of 50~000 points.

In more detail, we set up the standard, dense $\phi$-model with 3 hidden layers with 64 nodes each. The network is fully connected, has GELU activation function, Adam optimizer with 0.001 learning rate, and the weights of the custom losses of the $\cymetric$ package take default values. We use the default (\texttt{glorot\_uniform}) initialization of \texttt{TensorFlow} for the trainable weights of the neural networks. These hyperparameters were chosen after an initial optimization study on the standard $\phi$-model, and were kept constant across the invariant networks. The rationale for this choice was to isolate the effects of the invariant layers from effects of hyperparameter tuning. We leave a more extensive optimization study of the hyperparameters of invariant architectures for the future.

\subsection{Fermat quintic}
\label{sec:fermatexp}

Working on the Fermat quintic, we compare the  standard $\phi$-model  with networks where one, or several, non-trainable layers  are added.  
The added non-trainable layers correspond to the spectral layer which, as described in section \ref{sec:spectral}, maps the data to a form which is invariant under homogeneous rescalings. We also explore $G$-canonicalization layers that enforce invariance under homogeneous rescalings, permutations and scaling by roots of unity, as described in \ref{sec:hom}-\ref{sec:root}. 

\begin{figure}
    \centering
    \hspace*{-1.5cm}
\includegraphics[width=1.2\textwidth]{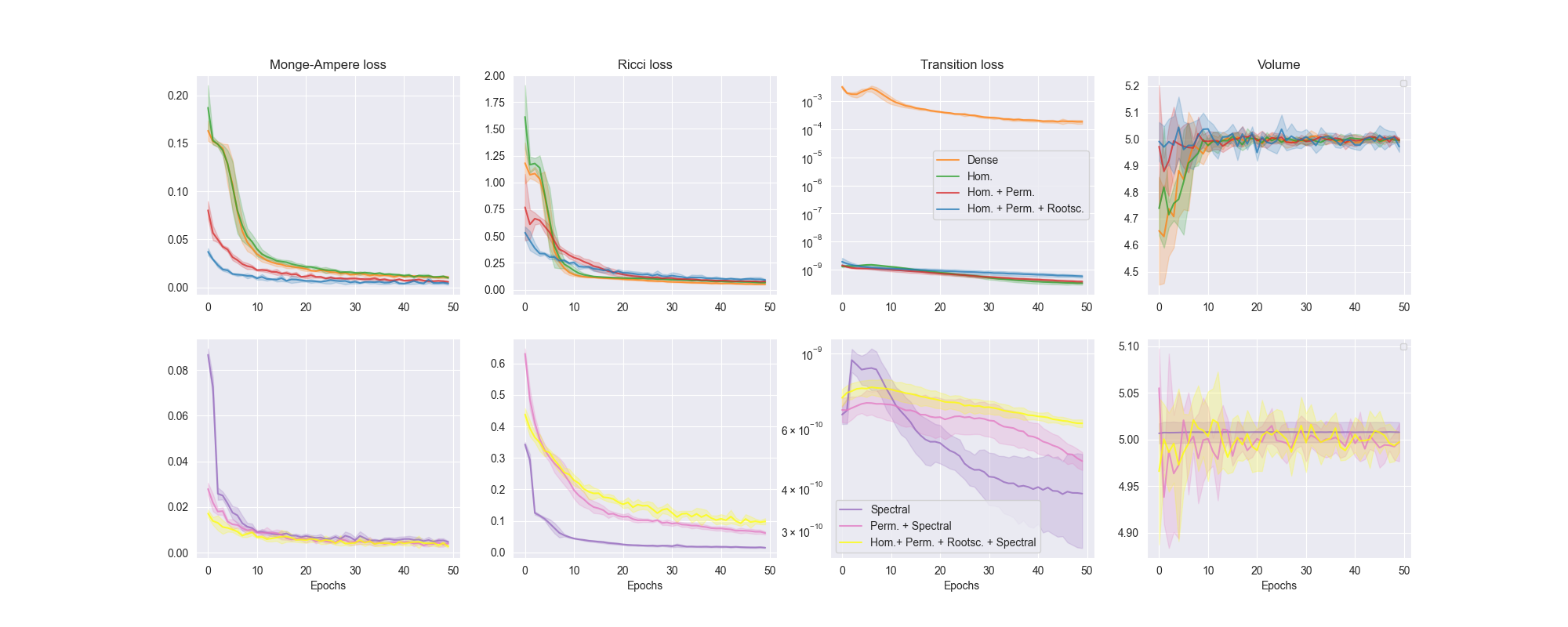}
    \caption{ Performance of different ML models on the Fermat quintic. From left to right: Monge--Amp\`ere, Ricci, and transition loss evolution on  validation data; volume computed with trained metric on all data (ground truth is 5). The upper (lower) row shows dense (spectral) networks with and without G-canonicalization layers. Solid line is average over 5 experiments, lighter color bands show the standard deviation.}
    \label{fig:fermat-5datasets}
\end{figure}

We have performed a number of experiments on this geometry, and report the results from training the models on 90~000 points, over 50 epochs  in figure~\ref{fig:fermat-5datasets} (we have checked that the conclusions are reproduced in 300-epochs experiments). Since training and validation plots look very similar, we only report the latter.\footnote{In more detail: the MA and transition losses on training data look very similar to the shown validation plots, showing that there is no problem with overfitting. The Kclass loss evolution is consistent with the Volume plot on validation data. K\"ahler and Ricci loss are disabled during training.} Thus, figure~\ref{fig:fermat-5datasets} shows the evolution, over the training epochs, of Monge-Amp\`{e}re, Ricci and transition loss on a validation set consisting of 10~000 points (the K\"ahler loss evaluates to $\mathcal{O}(10^{-16})$, and is omitted from the plot). We also track the evolution of the volume, as computed by the integral of the volume measure of the predicted metric (cf. appendix \ref{ap:integration}). Since the computation of the volume requires an integral, it is performed over the full data set, i.e. 100~0000 points. 

As is clear from the MA loss (leftmost in figure \ref{fig:fermat-5datasets}), the spectral layer and canonicalization layers all have a positive impact on the loss. This is clearly visible at initialization of the network. After training, the effect of the invariant layers on the MA loss is still positive, albeit less pronounced. This is in accord with previous observations in the literature indicating that the standard $\phi$-model is fairly good at learning symmetries.

When examining the Ricci loss, we notice a number of interesting features. First, the spectral layer has a larger positive impact on performance, both at initialization and during training. {Second, while the canonicalization layers may have a positive impact on the Ricci loss at initialization of standard $\phi$-models, the effect becomes negative after 10 epochs of training. While this discrepancy decreases with training, the standard $\phi$-model has a slightly lower Ricci loss on validation than its canonicalized counterparts. For spectral models, adding canonicalized layers has a negative effect on the Ricci loss. We interpret this as an effect of a complicated loss landscape where optimizing with respect to the  MA loss does not necessarily improve the Ricci loss (despite the fact that both of these losses are zero for the true Ricci-flat metric). Similar effects have been observed when learning CY metrics on more complicated CICY manifolds \cite{Schneider:2022ssn}, and may here be related to the fact that the symmetry group has fixed points, that may affect the computation of Ricci loss. As suggested in   ref.~\cite{Schneider:2022ssn}, by making a different hyperparameter choice, so that both MA and Ricci losses are used to guide the training of the ML model, performance may be improved; we leave such studies for the future. } 

The third subplot in figure \ref{fig:fermat-5datasets} shows the transition loss: this is $\mathcal{O}(10^{-3})$ for the dense network, and 6 orders of magnitude lower for the spectral network, and any network that includes a homogeneous canonicalization layer. This is in accordance with expectations, and indicates that $\mathbb{C}^*$-invariance is achieved.

Finally, the volumes computed with the trained models all converge to the ground truth upon training (within numerical errors). The spectral layer gives a constant volume under training, which shows  that this model indeed predicts a globally defined $\phi$ at any time during training. For the other architectures, there is more variance, so the same claim cannot be made. The standard, dense $\phi$-model exhibits most variation, and imposing $G$-canonicalization has a positive effect on the constancy of the volume. 

In table \ref{tab:fermat}, we list the $\sigma$ and $\mathcal{R}$ measures, which correspond to the integrated MA and Ricci loss respectively (see appendix \ref{ap:measures}), for the fully trained models as well as the measures for the Fubini-Study metric as a benchmark. We can see that adding the canonicalization layers has a positive effect on the $\sigma$ measure. Concatenating the homogeneous, permutation and rootscaling canonicalization layers improves the $\sigma$ measure by two orders of magnitude compared to the Fubini-Study metric, and by a factor of 2 compared to the standard, dense $\phi$-model. Similar improvement is observed for the spectral layer. For comparison, the $\sigma$ measure of the best performing experiments in table \ref{tab:fermat} is $0.0027$, which would correspond to  $k\approx 33$ with Donaldson's algorithm, and $k\approx 5$ with Headrick-Nassar's functional minimization.\footnote{This comparison uses the observations of $\sigma$ measures for $k\in[1,12]$, which shows that the $\sigma$ measure decreases polynomially, like $\sigma \sim 3.1 k^{-2}-4.2 k^{-3}$, for Donaldson's balanced metric \cite{Douglas:2006rr}, and exponentially, like $\sigma \sim 0.17 \times 2.2^{-k}$,  for Headrick-Nassar's optimal metric \cite{Headrick:2009jz}. } 
On the $\mathcal{R}$ measure, the spectral network performs the best. Any addition of canonicalization layers worsens the accuracy of the prediction, in particular when combined with the spectral layer. This is in accordance with the observed Ricci loss values on validation data, as discussed above.

\begin{table}[]
    \centering
    \begin{tabular}{|r|c|c|}
    \hline
       \textbf{Model}  & \textbf{$\sigma$ measure} & \textbf{$\mathcal{R}$ measure} \\
         \hline
       FS metric  & 0.3720 & 0.6919  \\
       Dense  & 0.0106 $\pm$ 0.0011 & 0.0502 $\pm$ 0.0040  \\
       Hom.  & 0.0100 $\pm$ 0.0006 & 0.0645 $\pm$ 0.0052  \\
       Hom. + Perm.  & 0.0058 $\pm$ 0.0007 & 0.0754 $\pm$ 0.0066  \\
       Hom. + Perm. + Rootsc.  & 0.0040 $\pm$ 0.0012 & 0.0891 $\pm$ 0.0084  \\
       Spectral  & 0.0047 $\pm$ 0.0004 &  0.0149 $\pm$ 0.0011 \\
       Perm. + Spectral  & 0.0040 $\pm$  0.0004 & 0.0600 $\pm$ 0.0049 \\
       Hom. + Perm. + Rootsc. + Spectral  & 0.0027 $\pm$  0.0005 & 0.0956 $\pm$ 0.0091 \\
       \hline
    \end{tabular}
    \caption{$\sigma$ and $\mathcal{R}$ measure for the Fermat quintic evaluated over a test set of 50~000 points. The values are rounded to four decimals and the error is given by the standard deviation of five experiments.}
    \label{tab:fermat}
\end{table}

\subsection{Tetraquadric}

In this example, we study the tetraquadric  defined by the configuration matrix 
\begin{align}
    \label{eq:bicubic}
    \left[
        \begin{array}{c|c}
        1 & 2 \\
        1 & 2 \\
        1 & 2 \\
        1 & 2
        \end{array}
    \right]
\end{align}
This CY 3-fold has $h^{(1,1)}=4$ and $h^{(2,1)}=68$. 
We pick a smooth member of this CICY family, defined by a degree (2,2,2,2)-polynomial $p$ of $\mathcal{M}$ which is invariant under the action of $S_2$, as defined in §\ref{sec:perm}. Concretely, we set
\begin{align*}
&p(x_0, x_1, y_0, y_1, z_0, z_1, w_0, w_1) :=\\
&2x_0^2 y_0^2 z_0^2 w_0^2 - 2x_0^2 y_0^2 z_0^2 w_0 w_1 + 2x_0 x_1 y_0^2 z_0^2 w_1^2 + x_1^2 y_0^2 z_0^2 w_0^2 + 2x_1^2 y_0^2 z_0^2 w_0 w_1 \\
&+ x_0^2 y_0 y_1 z_0^2 w_0^2 + x_0^2 y_0 y_1 z_0^2 w_1^2 - 2x_0 x_1 y_0 y_1 z_0^2 w_0 w_1 - x_0 x_1 y_0 y_1 z_0^2 w_1^2 \\
&+ x_1^2 y_0 y_1 z_0^2 w_0^2 - x_1^2 y_0 y_1 z_0^2 w_1^2 - 2x_0^2 y_1^2 z_0^2 w_0^2 - x_0^2 y_1^2 z_0^2 w_0 w_1 - 2x_0^2 y_1^2 z_0^2 w_1^2 \\
&- 2x_0 x_1 y_1^2 z_0^2 w_0^2 - 2x_0 x_1 y_1^2 z_0^2 w_0 w_1 + x_1^2 y_1^2 z_0^2 w_0^2 + 2x_1^2 y_1^2 z_0^2 w_0 w_1 + 2x_1^2 y_1^2 z_0^2 w_1^2 \\
&+ x_0^2 y_0^2 z_0 z_1 w_0^2 - 2x_0^2 y_0^2 z_0 z_1 w_1^2 - x_0 x_1 y_0^2 z_0 z_1 w_0^2 - 2x_0 x_1 y_0^2 z_0 z_1 w_0 w_1 \\
&+ x_0 x_1 y_0^2 z_0 z_1 w_1^2 - x_1^2 y_0^2 z_0 z_1 w_0^2 + x_1^2 y_0^2 z_0 z_1 w_0 w_1 - 2x_1^2 y_0^2 z_0 z_1 w_1^2 \\
&- x_0^2 y_0 y_1 z_0 z_1 w_0 w_1 + 2x_0^2 y_0 y_1 z_0 z_1 w_1^2 - x_0 x_1 y_0 y_1 z_0 z_1 w_0^2 - 2x_0 x_1 y_0 y_1 z_0 z_1 w_0 w_1 \\
&- x_0 x_1 y_0 y_1 z_0 z_1 w_1^2 + 2x_1^2 y_0 y_1 z_0 z_1 w_0^2 - x_1^2 y_0 y_1 z_0 z_1 w_0 w_1 - 2x_0^2 y_1^2 z_0 z_1 w_0^2 \\
&+ x_0^2 y_1^2 z_0 z_1 w_0 w_1 - x_0^2 y_1^2 z_0 z_1 w_1^2 + x_0 x_1 y_1^2 z_0 z_1 w_0^2 - 2x_0 x_1 y_1^2 z_0 z_1 w_0 w_1 \\
&- x_0 x_1 y_1^2 z_0 z_1 w_1^2 - 2x_1^2 y_1^2 z_0 z_1 w_0^2 + x_1^2 y_1^2 z_0 z_1 w_1^2 + 2x_0^2 y_0^2 z_1^2 w_0^2 + 2x_0^2 y_0^2 z_1^2 w_0 w_1 \\
&+ x_0^2 y_0^2 z_1^2 w_1^2 - 2x_0 x_1 y_0^2 z_1^2 w_0 w_1 - 2x_0 x_1 y_0^2 z_1^2 w_1^2 - 2x_1^2 y_0^2 z_1^2 w_0^2 \\
&- x_1^2 y_0^2 z_1^2 w_0 w_1 - 2x_1^2 y_0^2 z_1^2 w_1^2 - x_0^2 y_0 y_1 z_1^2 w_0^2 + x_0^2 y_0 y_1 z_1^2 w_1^2 \\
&- x_0 x_1 y_0 y_1 z_1^2 w_0^2 - 2x_0 x_1 y_0 y_1 z_1^2 w_0 w_1 + x_1^2 y_0 y_1 z_1^2 w_0^2 + x_1^2 y_0 y_1 z_1^2 w_1^2 \\
&+ 2x_0^2 y_1^2 z_1^2 w_0 w_1 + x_0^2 y_1^2 z_1^2 w_1^2 + 2x_0 x_1 y_1^2 z_1^2 w_0^2 - 2x_1^2 y_1^2 z_1^2 w_0 w_1 + 2x_1^2 y_1^2 z_1^2 w_1^2 \; .
\end{align*}
Note that this instance of the tetraquadric has no root scaling invariance. We project to the fundamental domain associated to $S_2$ as defined in eq.~(\ref{eq:FDtwoambient}). For the homogenous rescalings, we consider each ambient space sperately and project into the fundamental domain as in eq.~(\ref{eq:homcan_map}). The spectral layer is realized as describe in section \ref{sec:spectral} for a product of ambient spaces.

\begin{figure}
    \centering
    \hspace*{-1.5cm}
\includegraphics[width=1.2\textwidth]{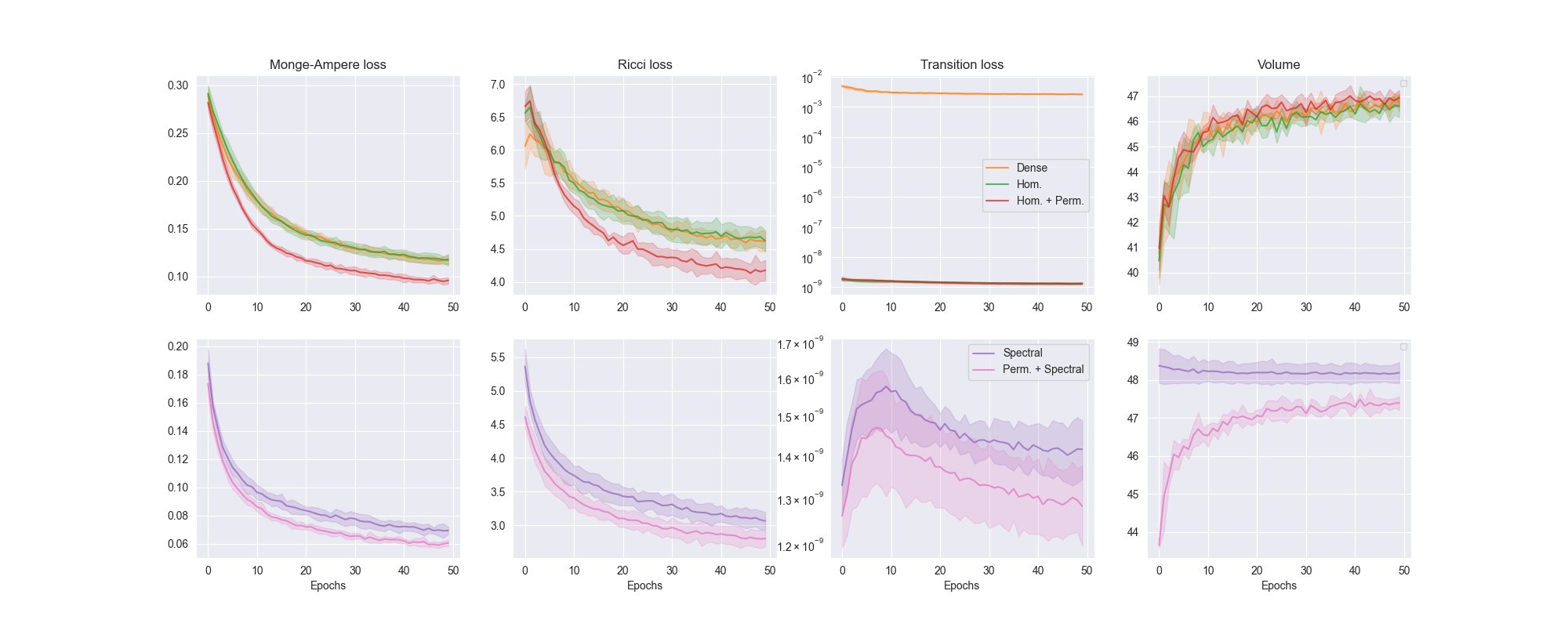}
    \caption{Performance of different ML models on a $S_2$ invariant tetraquadric CY 3-fold. From left to right: Monge--Amp\`ere, Ricci, and transition loss evolution on  validation data; volume computed with trained metric on all data (ground truth is 48). The upper (lower) row shows dense (spectral) networks with and without G-canonicalization layers.}
    \label{fig:tetra_losses}
\end{figure}

We perform analogous experiments as in the previous section, and report the validation results for five experiments trained over 90~000 points for 50 epochs in figure~\ref{fig:tetra_losses}. As expected from the previous experiment, adding the homogenous canonicalization layer leads to an almost vanishing transition loss while it does not improve the performance on the Monge--Amp\`ere and Ricci loss. However, adding the permutation canonicalization layer to the architecture improves performance both on Monge--Amp\`ere and Ricci loss. The spectral layer performs slightly better than the previous configuration, while the combination of permutation canonicalization and spectral layer performs best. 

\begin{table}[]
    \centering
    \begin{tabular}{|r|c|c|}
    \hline
       \textbf{Model}  & \textbf{$\sigma$ measure} & \textbf{$\mathcal{R}$ measure} \\
         \hline
       FS metric  & 0.4393 & 7.0902  \\
       Dense  & 0.1168 $\pm$ 0.0020 & 4.5680 $\pm$ 0.0823  \\
       Hom.  & 0.1183 $\pm$ 0.0017 & 4.5906 $\pm$ 0.0481  \\
       Hom. + Perm.  & 0.0960 $\pm$ 0.0018 & 4.1339 $\pm$  0.0882  \\
       Spectral  & 0.0665 $\pm$ 0.0002 &  2.9522 $\pm$ 0.0373 \\
       Perm. + Spectral  & 0.0582 $\pm$  0.0007 & 2.6980 $\pm$ 0.0326 \\
       \hline
    \end{tabular}
    \caption{$\sigma$ and $\mathcal{R}$ measure for the tetraquadric evaluated over a test set of 50~000 points. The values are rounded to four decimals and the error is given by the standard deviation of five experiments.}
    \label{tab:tetraquadric}
\end{table}

For the volume, we make the following observations. The spectral layer is most stable over all epochs and follows, within a numerical error, the ground truth of $48$  while all other combinations show higher variance and seem to converge to a slightly shifted value. As discussed in section \ref{sec:fundproj}, the variance could be explained due to the fact that the fundamental domain projections introduce discontinuities which might skew calculations, like the volume, which require integration over the entire manifold. However, the shifted volume prediction on the tetraquadric is not observed in the other experiments we have performed. The explanation behind this difference is likely that   $h^{(1,1)}>1$ for the tetraquadric, while it is 1 for the other CY manifolds, and that using a loss function to fix the K\"ahler class during training is more challenging on manifolds with higher $h^{(1,1)}$. 

In table \ref{tab:tetraquadric}, we compute the $\sigma$ and $\mathcal{R}$ measures for the different models. For the $\sigma$ measure, we observe a similar trend as for the Fermat quintic. Namely, adding any combination of canonicalization layers\footnote{Adding only the homogeneous canicalization layer does not improve the $\sigma$ measure. The main purpose of this layer is to reduce the transition loss.} has a positive effect on the measure. The spectral layer shows the most significant improvement, while the combination of permutation canonicalization layer and spectral layer perform the best. The $\mathcal{R}$ measure shows a different trend as for the Fermat quintic. Adding the permutation canonicalization layer improves the measure slightly, while again the spectral layer shows the most significant improvement. Again, the combination of permutation canonicalization layer and spectral layer performs the best.

\subsection{A 2-polynomial CICY}

\begin{figure}
    \centering
    \hspace*{-2cm}
\includegraphics[width=1.2\textwidth]{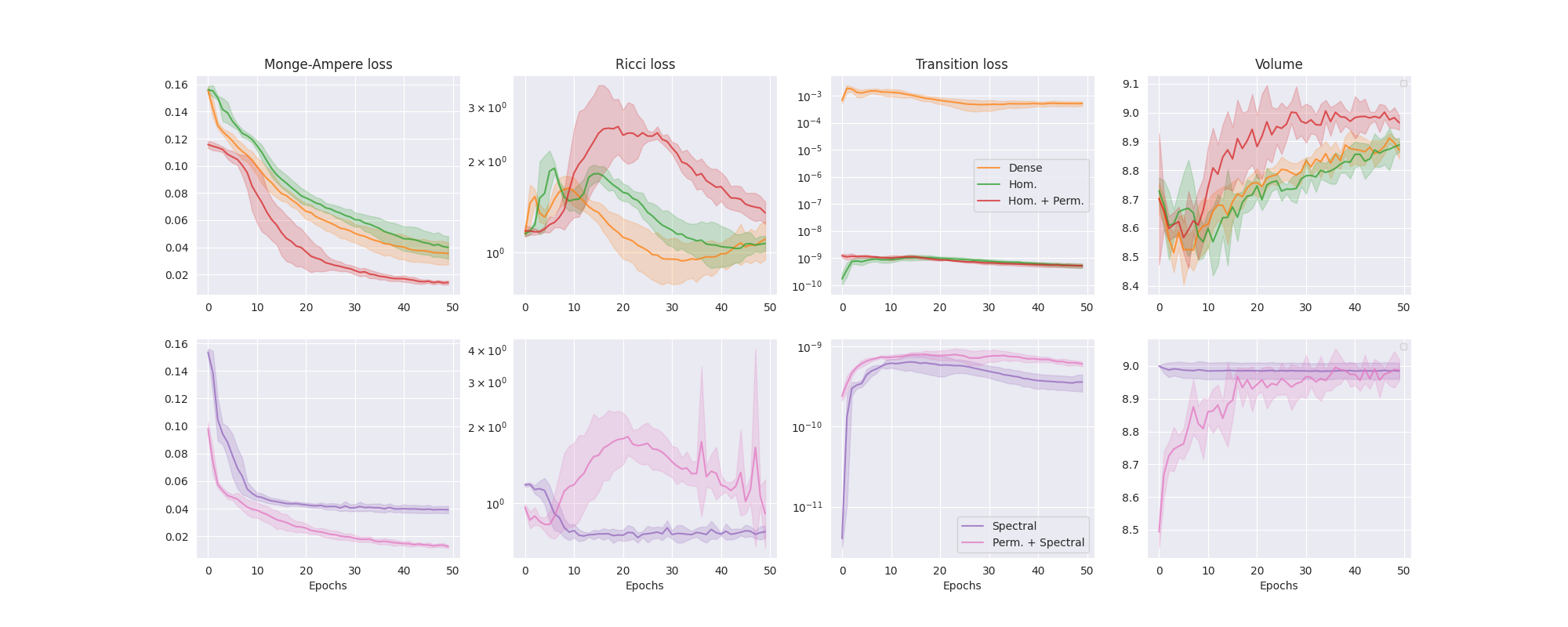}
    \caption{Performance of different ML models on a $[5|3\  3]$ CY 3-fold. From left to right: Monge--Amp\`ere, Ricci, and transition loss evolution on  validation data; volume computed with trained metric on all data (ground truth is 9). The upper (lower) row shows dense (spectral) networks with and without G-canonicalization layers.}
\label{fig:533_runs}
\end{figure}

We look at the  CICY $[5|3\ 3]$, where the two defining polynomials are given by (cf. \cite{Butbaia:2024tje})
\[
\begin{split}
    p(x_0, x_1, x_2, x_3, x_4, x_5) &:= x_0^3 + x_1^3 + x_2^3 - 3x_3x_4x_5,\\
    q(x_0, x_1, x_2, x_3, x_4, x_5) &:= x_3^3 + x_4^3 + x_5^3 - 3x_0x_1x_2.
\end{split}
\]

This CICY has a discrete group $\mathbb{Z}_2 (S_3 \times S_3)$ of permutation symmetries, where the non-trivial element in $\mathbb{Z}_2$ acts by sending  $[x_0: x_1: x_2: x_3: x_4: x_5]$ to $[x_3: x_4: x_5: x_0: x_1: x_2] $.
while the first $S_3$ factor permutes the coordinates $(x_0, x_1, x_2)$ and the second $S_3$ factor permutes the coordinates $(x_3, x_4, x_5)$. {The group does not act freely.}

We implement a canonicalization for the action of $\mathbb{Z}_2 (S_3 \times S_3)$ in two steps. In the first step, we reorder the coordinates in the first triple and the second triple in descending order with respect to norm, separately.
\[
\begin{split}
      &[ x_0 : x_1 : x_2 : x_3 : x_4 : x_5 ] \mapsto [x_{\sigma(0)} : x_{\sigma(1)} : x_{\sigma(2)} : x_{\tau(3)} : x_{\tau(4)} : x_{\tau(5)}]\\
      &\text{where}\ (\sigma , \tau) \in S_3 \times S_3,\ |x_{\sigma(0)}| \geq |x_{\sigma(1)}| \geq |x_{\sigma(2)}|,\   |x_{\tau(3)}| \geq |x_{\tau(4)}| \geq  |x_{\tau(5)}|.
\end{split}
\]
In the second step, we shift such that the triple with highest sum of norms is first.
We combine these two steps in one permutation canonicalization layer.

We perform experiments similar to the previous sections where we combine homogeneous and permutation canonicalization layers, and add them to the dense and spectral $\phi$-models. In figure \ref{fig:533_runs}, we see that the homogeneous canonicalization layer leads to an almost vanishing transition loss when combined with the dense network.
Regarding volume computation, we can see that adding the permutation canonicalization layer helps the dense network to converge faster to the correct volume, but it worsens the computation for the spectral network.
With respect to the Monge--Amp\`ere loss, the permutation canonicalization layer improves the performance of both networks by a factor of 2.
In contrast to the previous experiments, the Ricci loss shows different trends from the Monge--Amp\`ere loss.
Only the spectral network reaches a clear plateau value after 10 epochs, while the other networks seems to oscillate.
The permutation canonicalized models has an initial rise followed by a steep descent; {since the model has not settled, further training might improve the Ricci loss, and training against the Ricci loss, in addition to the MA loss, might also improve performance, as discussed for the Fermat quintic experiments.}

In table \ref{tab:533}, we compute the $\sigma$ and $\mathcal{R}$ measures for the different models. For the $\sigma$ measure we observe almost identical performance for all architectures. {The $\mathcal{R}$ measure does not improve for any architecture except the spectral model, which is in accordance with the Ricci loss validation plots. The explanation is likely that the loss landscape for this geometry is hard to navigate, and that models get stuck in a local minimum.} This is supported by the observation that even the standard, dense $\phi$-model shows poorer performance {on the $\mathcal{R}$ measure} than on the other geometries.

\begin{table}[]
    \centering
    \begin{tabular}{|r|c|c|}
    \hline
       \textbf{Model}  & \textbf{$\sigma$ measure} & \textbf{$\mathcal{R}$ measure} \\
         \hline
       FS metric  & 0.2620 & 1.0907  \\
       Dense  & 0.1971 $\pm$ 0.0015 & 0.9796 $\pm$ 0.1174  \\
       Hom.  & 0.1974 $\pm$ 0.0020 & 1.0003 $\pm$ 0.0463  \\
       Hom. + Perm.  & 0.1983 $\pm$ 0.0006 & 1.2831 $\pm$  0.1380  \\
       Spectral  & 0.1944 $\pm$ 0.0002 &  0.7203 $\pm$ 0.0291 \\
       Perm. + Spectral  & 0.1981 $\pm$  0.0007 & 2.2978 $\pm$ 2.5123 \\
       \hline
    \end{tabular}
    \caption{$\sigma$ and $\mathcal{R}$ measure for the $[5|3\  3]$ CY 3-fold evaluated over a test set of 50~000 points. The values are rounded to four decimals and the error is given by the standard deviation of five experiments.}
    \label{tab:533}
\end{table}

\section{Ricci-flat metrics on non simply connected CY manifolds}
\label{sec:freecanon}

In heterotic string compactifications, phenomenologically interesting particle physics models are often based on CY manifolds that are not simply connected. Such manifolds admit Wilson lines, which allow us to break the heterotic $E_8$ group to the gauge group of the Standard Model, and achieve the correct number of particle generations; this idea is pedagogically explained in \cite{green2012superstring}, and dates back to \cite{Candelas:1985en,WITTEN198575,1985PhLB..158...33B}. There is a large literature on the use of discrete symmetries in heterotic model building, and the reader is referred to the papers cited above for more discussion on this topic. In this section, we will follow \cite{Braun:2007sn} and study the Ricci-flat metric on a quotient $Q$ of a smooth, 5-parameter quintic $X$ by a freely acting $\mathbb{Z}_5^2$ group.    As explained e.g. in \cite{green2012superstring,Braun:2007sn}, $Q$ has Euler characteristic -8, and leads to a particle physics model with 4 generations, so we will refer to this geometry as a 4-generation quotient of the quintic. 

In more detail, we will study the 5-parameter quintic CY specified, in $\mathbb{CP}^4$, by the zero locus of
\begin{equation}
\begin{split}
       p = &\, (z_0^5 + z_1^5 + z_2^5 + z_3^5 + z_4^5) \\  
    &+\, \psi_1 z_0 z_1 z_2 z_3 z_4  \\
    &+\, \psi_2 (z_0^3 z_1 z_4 + z_0 z_1^3 z_2+z_0 z_3 z_4^3+z_1 z_2^3 z_3 + z_2 z_3^3 z_4) \\
    &+\, \psi_3 (z_0^2 z_1 z_2^2 + z_1^2 z_2 z_3^2+z_2^2 z_3 z_4^2+z_3^2 z_4 z_0^2 + z_4^2 z_0 z_1^2) \\
    &+\, \psi_4 (z_0^2 z_1^2 z_2 + z_1^2 z_2^2 z_3+z_2^2 z_3^2 z_4+z_3^2 z_4^2 z_0 + z_4^2 z_0^2 z_1) \\
    &+\, \psi_5 (z_0^3 z_2 z_3 +  z_1^3 z_3 z_4+z_2^3 z_4 z_0+z_3^3 z_0 z_1 + z_4^3 z_1 z_2) \, .
\end{split}
\end{equation}
We pick a generic instance of this family of CY manifolds by drawing $\psi_i$ from $\mathcal{N}(0, 1)$.

As is shown in ref.~\cite{Braun:2007sn}, this 5-parameter quintic $X$  is invariant under a  group $G= \mathbb{Z}_5^2$ generated by 
\begin{equation}
g_1 =
    \begin{pmatrix}
        0 & 0 & 0 & 0 & 1 \\
        1 & 0 & 0 & 0 & 0 \\
        0 & 1 & 0 & 0 & 0 \\
        0 & 0 & 1 & 0 & 0 \\
        0 & 0 & 0 & 1 & 0 \\
    \end{pmatrix}
    \; , \; 
g_2 =
    \begin{pmatrix}
        1 & 0 & 0 & 0 & 0 \\
        0 & e^{\frac{2\pi i}{5}} & 0 & 0 & 0 \\
        0 & 0 & e^{\frac{4\pi i}{5}} & 0 & 0 \\
        0 & 0 & 0 & e^{\frac{6\pi i}{5}} & 0 \\
        0 & 0 & 0 & 0 & e^{\frac{8\pi i}{5}} \\
    \end{pmatrix}    
\end{equation}
Moreover, the group $G$ is shown to act freely on $X$, so that the quotient $Q=X/G$ is smooth.

In ref.~\cite{Braun:2007sn}, Braun and collaborators obtain an approximation for the Ricci-flat CY metric on $Q$ using the Donaldson algorithm. As was mentioned in the introduction of this paper, this algorithm computes the CY metric via an expansion in a spectral basis, consisting of degree-$k$ homogeneous polynomials, using a fixed point iteration scheme to determine the coefficients of the expansion. The accuracy of metric improves as $k \to \infty$. The discrete symmetry is imposed as a restriction on the spectral basis, and results in a slower growth of the basis dimension. The invariant spectral basis can thus be constructed for larger $k$ than the non-invariant one; consequently more accurate approximations for the metric can be achieved.

However,  it is somewhat non-trivial to determine the basis of sections on $Q$, i.e. identify the linearly independent $G$-invariant sections on $X$ that restrict to $Q$, already for moderate values of $k$. There are two subtleties: first, polynomials of the form $q_i p$, where $p=0$ defines the CY hypersurface, evaluate to zero on the CY, and must be removed from the invariant basis. Second, the generators $g_1, g_2$ only commute up to homogeneous rescaling of the coordinates, which makes the construction of invariant polynomials non-trivial. 
Braun and collaborators rely on algorithms of invariant theory,  implemented using SINGULAR \cite{DGPS}, to identify the basis, and explore these methods up to $k =25$. After running the Donaldson algorithm on a a sample with $10^6$ points, they find that the $\sigma$ measure evaluates to $0.015$ for $k =25$ on an unseen test set of 20~000 points.\footnote{
When comparing the performance against number of basis elements, they find that the $G$-invariant metric approximation  has smaller $\sigma$ error than a metric computed using the unrestricted spectral basis.}

In this section, we revisit the computation of the Ricci-flat CY metric on $Q$ using the methods introduced in section \ref{sec:fundproj}. We thus encode the $G$-invariance as non-trainable canonicalization layers in the $\phi$-model of \cymetric, and compare the performance of these invariant ML models with the standard dense counterpart, as well as spectral networks. 
{Notice that the predicted metric using the above canonicalization layers is on the covering space $X$, rather than being on the quotient $Q$: 
the metric is invariant under the action of $G$, and so  constant along the fibers, but it is expressed in coordinate patches on  $X$.
Subsequently, in order to obtain the metric on the quotient, one needs to pushforward the learned metric via the projection from  $X$ to $Q$.
In this work, we do not execute this last step.}

\begin{figure}
    \centering
    \hspace*{-1.5cm}
\includegraphics[width=1.2\textwidth]{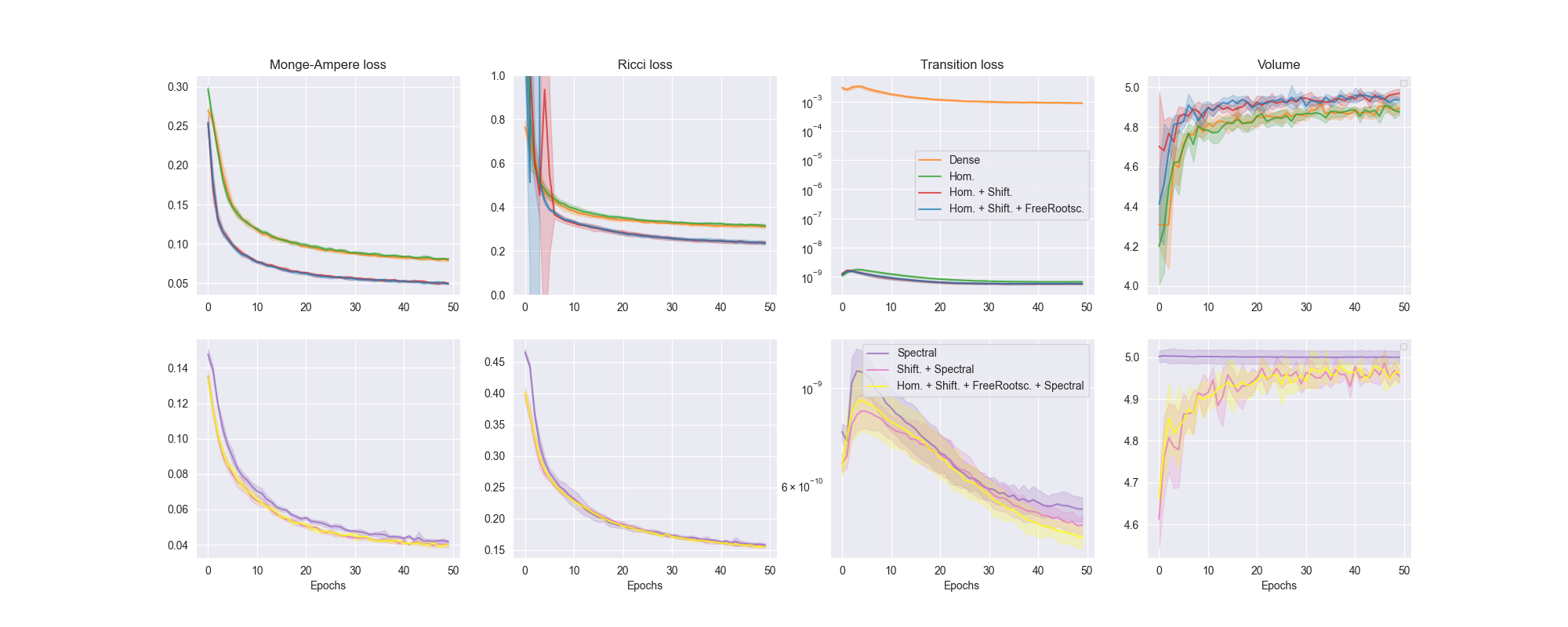}
    \caption{Performance of different ML models on a  4-generation quintic. From left to right: Monge--Amp\`ere, Ricci, and transition loss evolution on  validation data; volume computed with trained metric on all data (ground truth is 5). The upper (lower) row shows dense (spectral) networks with and without G-canonicalization layers.}   \label{fig:4Gquintic}
\end{figure}

We implement the canonicalization layers as follows. 
Let $G_i \subset G$ be the subgroup generated by the generator $g_i$ for $i=1,2$. 
Then, by projecting on the fundamental domains of $G_1$ and $G_2$ consecutively, we construct an invariant function on $X$ w.r.t $G$ or, in other words, a function on the smooth quotient $Q = X/G$. A technical complication, compared to the Fermat quintic canonicalization layers of Sec.~\ref{sec:fermatexp}, is that we need to address the non-commutativity of $G_1$ and $G_2$. This can be resolved by inserting the homogeneous canonicalization layer before the other $G$-canonicalization layers; $g_1$ and $g_2$ commute up to homogeneous rescalings.

The fundamental domains of $G_1$ and $G_2$ are given by
\begin{align}
    \mathcal{F}_1 &:= \{[z_0:z_1:z_2:z_3:z_4] \in \proj{4} |\ \forall i : |z_0| \geq |z_i|  \},\\
    \mathcal{F}_2 &:= \{[z_0:z_1:z_2:z_3:z_4] \in \proj{4} |\ 0\leq \arg(z_0) \leq 2\pi/5 \}.
\end{align}
To project on the fundamental domains $\mathcal{F}_1$ and $\mathcal{F}_2$, we implement two canonicalization layers that act on the homogeneous coordinates $[z_0:z_1:z_2:z_3:z_4]$ of a point on $X$. The first layer returns the coordinate tuple shifted to the left such that the  highest norm entry takes the first spot, 
\begin{align*}
    h_{\rm shift}: \, (z_0,z_1,z_2,z_3,z_4) \mapsto &(z_i,z_{(i+1\mod5)},z_{(i+2 \mod 5)},z_{(i+3 \mod 5)},z_{(i+4 \mod 5)}),\\
    &\ \forall j:\ |z_i|\geq |z_j|.
\end{align*}
For the second canonicalization layer, we find the fifth root of unity $a$ such that $0\leq\arg(az_0)\leq 2\pi/5$ and then multiply the rest of the coordinates with successive powers of $a$ as follows,
\[
h_{\rm FreeRootSc}: \,(z_0,z_1,z_2,z_3,z_4) \mapsto (az_0,a^2z_1,a^3z_2,a^4z_3,z_4).
\]

Our experiments run over 50 epochs with the same settings as in section \ref{sec:exp}. The results on validation data are shown in figure~\ref{fig:4Gquintic} (we again omit the training plots which are more or less identical to the validation plots).  Just as observed in other experiments, including a homogeneous canonicalization layer, or a spectral layer, significantly reduces the transition loss. The validation plots of Monge-Amp\`{e}re and Ricci loss shows that canonicalization with respect to $G_1, G_2$ reduces the losses, albeit not by a large factor.\footnote{The spikes in the Ricci loss in the second plot of the top row in figure~\ref{fig:4Gquintic} are transient numerical fluctuations.} The reduction of the loss is more pronounced when $G$-canonicalization is added to the dense model, than when similar additions are made in the spectral network. The addition of canonicalization layers also has a small positive effect on the volume for the dense network. For the spectral network, the effect on the volume prediction is instead negative at initialization. While this effect is subdued by training, it is an indication that the canonicalization layers introduce discontinuities, which affect global predictions of the model in a negative way.  

In table \ref{tab:free_quintic}, we compute the $\sigma$ and $\mathcal{R}$ measure for all the models. We can observe an improvement on the $\sigma$ measure by adding any combination of canonicalization layers to the standard, dense $\phi$-model.\footnote{Again, except for the case in which only the homogeneous canonicalization layer is added, which performs on par with the dense network.}  The spectral network performs better than any model based on the standard model, and combining the spectral layer with canonicalization layers gives the lowest $\sigma$ measure. For comparison, the error measure values obtained by the fully canonicalized dense (spectral) networks correspond to the value obtained with $k$ about $15$ ($20$) using Donaldson's algorithm \cite{Braun:2007sn}.  For the $\mathcal{R}$ measure, we observe the same trend as for the $\sigma$ measure.

Summing up the results of this experiment, we note that the canonicalization approach is easily implemented on this quintic, and that the layers combine nicely with both the dense  and spectral $\phi$-models. Combining this with the pushforward, via the projection from  $X$ to $Q$, we can thus learn the metric on quotient CY manifolds using these machine learning techniques.   Furthermore, the accuracy of the prediction, as measured by the $\sigma$ and $\mathcal{R}$ measures, improves, and compares well to the results of ref.~\cite{Braun:2007sn}. We leave a more rigorous study of the performance gains that can be attained with these models, compared to previous work, for future hyperparameter optimization studies. What we can conclude already now is that, since there is no need to construct the $G$-invariant spectral basis, the $G$-canonicalization layers seem to provide a technical simplification compared to 
spectral methods used in ref.~\cite{Braun:2007sn}

\begin{table}[]
    \centering
    \begin{tabular}{|r|c|c|}
    \hline
       \textbf{Model}  & \textbf{$\sigma$ measure} & \textbf{$\mathcal{R}$ measure} \\
         \hline
       FS metric  & 0.4006 & 0.7876  \\
       Dense  & 0.0802 $\pm$ 0.0007 & 0.3070 $\pm$ 0.0048  \\
       Hom.  & 0.0818 $\pm$ 0.0011 &  0.3115 $\pm$ 0.0053  \\
       Hom. + Shift  & 0.0512 $\pm$  0.0009 &  0.2374 $\pm$ 0.0039  \\
       Hom. + Shift + FreeRootsc.  & 0.0523 $\pm$  0.0012 & 0.2396 $\pm$ 0.0062  \\
       Spectral  & 0.0419 $\pm$ 0.0004 & 0.1571 $\pm$ 0.0026 \\
       Shift + Spectral  & 0.0404 $\pm$  0.0014 & 0.1570 $\pm$ 0.0032 \\
       Hom. + Shift + FreeRootsc. + Spectral  & 0.0399 $\pm$  0.0011 & 0.1549 $\pm$ 0.0017 \\
       \hline
    \end{tabular}
    \caption{$\sigma$ and $\mathcal{R}$ measure for the 4-generation quotient of the quintic evaluated over a test set of 50~000 points. The values are rounded to four decimals and the error is given by the standard deviation of five experiments.}
    \label{tab:free_quintic}
\end{table}

\section{Conclusion}
\label{sec:conclude}
In this paper, we have addressed the question: how do we encode symmetries in a machine learned Ricci-flat CY metric? From a geometric point of view, CY manifolds may have discrete symmetries, with interesting consequences in mathematics and physics. From a machine learning perspective, these symmetries are manifest as invariances, under group transformations, of the input data of machine learning models designed to predict CY metrics. By constructing invariant ML models, one expects better performance and accuracy. A clear application of these invariant models are in string phenomenology, for example the computation of heterotic Yukawa couplings, which was recently explored using machine learning tools in \cite{Butbaia:2024tje,Constantin:2024yxh}. 

Our work is inspired by recent developments in the ML literature, that advocate to achieve model invariance by a pre-processing step, where the input data is projected into a fundamental domain of the relevant group \cite{aslan2022group,kaba2023equivariance}. We have used an approach where non-trainable layers perform these projections; collectively, we refer to these layers as canonicalization layers. This is a  modular implementation, and the layers can easily be exchanged to match different settings. While we have limited our studies to the \cymetric package \cite{Larfors:2021pbb,Larfors:2022nep}, we expect similar canonicalization layers can be added with ease in other ML packages for CY geometry.

We have performed experiments on different CICY manifolds, comparing the performance of the standard $\phi$-model of the \cymetric package with and without canonicalization layers, {but otherwise identical hyperparameter settings. We have also compared  with the spectral network, originally proposed by Berglund and collaborators \cite{Berglund:2022gvm}, which may also be combined with canonicalization layers. We find that adding canonicalization layers increases performance on the losses that are used to train the networks. In particular, the increase} in performance on the transition loss is several orders of magnitudes for all architectures. {The volume, controlled by the K-class loss, converges faster if canonicalization layers are added to the standard $\phi$-model.  However, the volume seems to be captured best by the spectral networks without canonicalization. Being a global property, this could be explained by discontinuities, related to the boundaries of fundamental domains, that the canonicalization layers may introduce. These are not present in the spectral layer. }

The improvement on the Monge--Amp\`ere loss is noticeable but not very large. Combining the spectral layer with the canonicalization layers seems to additionally increase performance while it introduces variance in the global prediction of the volume. However, adding these custom layers can still be of use regarding performance, as in many cases the losses converge faster, allowing one to reduce the number of training epochs. In addition, it has not yet been tested how the canonicalized layers perform in a systematic hyperparameter optimization study.

The Ricci loss yields different results depending on which geometry we look at. For the tetraquartic, as well as the 4-generation quotient of the quintic, we are quotienting by a freely acting symmetry group and we observe that canonicalization improves the Ricci loss on validation data, under training governed by the MA loss. The $\mathcal{R}$ error measures also decrease for canonicalized networks, compared to the standard dense and spectral $\phi$-models. For the Fermat quintic and the $[5|3\  3]$ CY 3-fold, adding the canonicalization layers instead increases the Ricci loss on validation data, and leads to a higher Ricci measure on test data.  This could be explained by a more complicated loss landscape on these geometries, where training against the MA loss does not directly decrease the Ricci loss. The fact that the canonicalization layers  encode a symmetry group that does not act freely on these geometries may contribute to this complicated loss landscape. 

Furthermore, we would like to emphasize the advantage of the canonicalization layers following the philosophy of the Geometric Deep Learning program. The unrestricted models, consisting solely of dense layers, seem to already learn the symmetry of the data.  This is in accord with previous observations in the literature, see e.g. \cite{Anderson:2023viv}, showing that the performance of the same dense network improves slightly on symmetric CY manifolds, compared to generic ones. However, adding the canonicalization layers creates an architecture that knows about the symmetries a priori, allowing for better interpretability of the model.

Moreover, using the canonicalization approach, we have taken steps towards learning  metrics of non-simply connected CY manifolds. Such manifolds, $Q$, always have a simply connected CY manifold $X$ as a covering space; the relation between the two spaces is $Q = X/\Gamma$, where $\Gamma$ is a freely acting discrete symmetry group. The canonicalization idea provides a simple way of imposing $\Gamma$-invariance on the learned Ricci-flat metric of $X$, and thus, by pushforward, we can obtain the Ricci-flat metric of $Q$. While the approximation of CY quotient metrics is already well-studied, using other techniques \cite{Braun:2007sn,Headrick:2009jz}, the simplicity of the canonicalization layers, and the fact that they can be combined with relative ease, makes for a practical method, that can easily be adapted to different CY spaces and groups. 

Looking ahead, there are many other directions yet to be explored. We have already mentioned that a thorough hyperparameter study might improve the performance of canonicalized architectures. Once this is done, it would be interesting to establish if scaling laws, as first observed in ref.~\cite{2020arXiv200108361K}, hold for these ML models, and use these laws to compare the efficiency of different architectures. 
One could also explore the alternative invariant architectures suggested in ref.~\cite{yarotsky:2022}, i.e. construct invariant layers using polynomial invariants. While we have performed some investigations using this approach, they indicated that such layers quickly become too wide to be practically useful, even for the discrete symmetries prevalent among CY manifolds. In addition, this approach lacks the modular aspect of the canonicalization layers. A related avenue is to construct invariant models based on the invariant spectral basis to higher orders in $k$.  This is similar to what has been attempted for algebraic CY metrics in the past, e.g.~\cite{Douglas:2006rr,Headrick:2009jz,Braun:2007sn}, so one would expect performance gains, even if we expect the construction of such invariant layers face similar computational subtleties as reported in \cite{Braun:2007sn}.

Finally, there are many interesting phenomenological applications of $G$-invariant ML models, such as the numerical computation of Hermitian Yang Mills connections \cite{Ashmore:2021rlc} and Yukawa couplings \cite{Butbaia:2024tje,Constantin:2024yxh}. We expect  that the canonicalization approach will be particularly relevant for heterotic compactifications with non-standard embedding of the gauge group; as shown in the examples explored in \cite{Constantin:2024yxh}, several neural nets must be trained to learn the Ricci-flat CY metric, Hermitian-Yang-Mills connection and harmonic representatives of certain vector bundle cohomologies. The $G$-invariance-imposing canonicalization layers can be applied to all of these networks; any gain in the computation cost achieved by the canonicalization layer will be multiplied by the number of networks to which it is applied.  We expect this to be of importance for  phenomenological studies.

\section*{Acknowledgement}
This work was in part supported by  the Wallenberg AI, Autonomous Systems and Software Program (WASP) funded by the Knut and Alice Wallenberg Foundation. Support was also obtained from Vetenskapsrådet, grant no.~2020-03230, and from Uppsala University's AI4Research center. The computations were partly enabled by resources provided by the National Academic Infrastructure for Supercomputing in Sweden (NAISS), partially funded by Vetenskapsrådet through grant agreement no.~2022-06725. ML thanks the Erwin Schrödinger Institute in Vienna, and the organisers of the  program {\it The Landscape vs. the Swampland}, for hospitality during the completion of this work. MW thanks Nordita for hospitality during the completion of this work.

We thank  Lara Anderson, Georgios Dimitroglou Rizell, Cristofero Fraser-Taliente, James Halverson, Thomas Harvey, Daniel Platt, Carlos Rodriguez, Fabian Ruehle, Robin Schneider, and Ayca Özcelikkale for valuable discussions and suggestions. We also thank Carlos Rodriguez and Robin Schneider for comments on a draft of this paper;  Fabian Ruehle for his continuous work in maintaining the \cymetric package; Muyang Liu and Martin Svärdsjö for  collaborations on related topics; and, finally, Daniel Persson and Jan Gerken for organizing the Mathy-AI Journal Club, as well as its attendants, for providing a lot of inspiration for this work.

\clearpage
\appendix

\section{CY geometry}
\label{ap:isometry}

In this section, we prove that symmetries of the defining polynomial of a CICY manifold are isometries of the Ricci-flat metric. This should be well-known, but we have not found the proof in the standard references on the subject, which motivates including this section.

Note that the notation in this section differs from other sections, where $J$ is used to define the K\"ahler form. Here, we will follow the prevalent notation of the mathematical literature and use $\omega$ for the K\"ahler form, and $J$ for a complex structure. 

We start with some preliminary definitions from §8.4 in \cite{Nakahara:2003nw} . 
Let $X$ be a complex manifold with a complex structure $J$.
Then a hermitian metric $g$ on $X$ is a Riemann metric that preserves the complex structure i.e $g(J\cdot, J\cdot) = g(\cdot, \cdot)$.
Given a hermitian metric $g$ on $X$, we can define a 2-form $\omega(\cdot, \cdot) := g(J\cdot, \cdot)$ on $X$ which is called the \emph{Kähler} form, and  $X$ is called \emph{Kähler} if its K\"ahler form is closed, i.e. $d\omega = 0$. 
Hence a K\"ahler manifold is equipped with three structure compatible structures $J, g$ and $\omega$.
A fact from K\"ahler geometry tells us that given any two of the three structures, the compatibility conditions between them completely determine the third one.
This immediately implies that if a given map between K\"ahler manifolds preserve any two of the structures then it preserves the third one.
Finally, a map $f: X\to X$ is holomorphic if and only if it satisfies the Riemann-Cauchy equation given by
\[
J \circ df = df \circ J.
\]

\begin{definition}
    A Calabi-Yau manifold $X$ is a compact K\"ahler manifold with a vanishing first Chern class, or equivalently with a nowhere-vanishing holomorphic top form. 
\end{definition}

Let $X$ be a Calabi-Yau manifold embedded in a product of complex projective spaces.
The Fubini-Study form restricts to a K\"ahler form $\omega_{FS}$ on $X$. 
The Calabi-Yau conjecture, or more precisely theorem, states that there exists a unique smooth function $\phi: X \to \mathbb{R}$, up to addition of a constant, such that the hermitian metric $g_{CY}$ associated with the form $\omega_{CY} := \omega_{FS} + i\del\delbar \phi$ is Ricci-flat.
In the following propositions, we fix $X$ to be a Calabi-Yau manifold embedded in a product of projective spaces.

\begin{proposition}\label{proposition: isometry_of_FS_implies_isometry_of_CY_flat}
    Let $f: X\to X$ be a biholomorphic isometry of $g_{FS}$ then $f$ is an isometry of the Ricci-flat metric $g_{CY}$ on $X$ i.e $f^* g_{CY} =g_{CY}$.
\end{proposition}
\begin{proof}
    The map $f$ preserves the Fubini-Study form $\omega_{FS}$ on $X$.
    Using this and the fact that the pullback via a holomorphic map commutes with the Dolbeault operators, we get
    \[
    f^*\omega_{CY} = f^*\omega_{FS} + if^*\del\delbar \phi = \omega_{FS} + i\del\delbar (\phi\circ f).
    \]

    Notice that $f^*g_{CY}$ is the associated hermitian metric with $ f^*\omega_{CY}$, and 
    \[
    Ric(f^*g_{CY}) = f^* Ric(g_{CY}) = 0 \, .
    \]
    Therefore, by the uniqueness part of the Calabi-Yau theorem, we have $f^* \omega_{CY} = \omega_{CY}$ and thus $f^* g_{CY} = g_{CY}$.
\end{proof}
\begin{proposition}
    Let $\phi: X \to \mathbb{R}$ be a smooth function that satisfies the equation $\omega_{CY} = \omega_{FS} + i\partial \Bar{\partial} \phi $. Assume that $\mathcal{S}$ is a finite group of biholomorphic symmetries of $\omega_{FS}$.
    Then $\phi$ is invariant under the action $\mathcal{S}$, i.e. for all $f\in \mathcal{S}$, we have $f^* \phi = \phi \circ f = \phi$.
\end{proposition}
\begin{proof}
    By Proposition \ref{proposition: isometry_of_FS_implies_isometry_of_CY_flat}, we have $\mathcal{S}$ is a group of isometries of $\omega_{CY}$.
    Consequently, for all $f\in \mathcal{S}$, we have
    \[
    \omega_{FS} + i\del\delbar \phi = \omega_{CY} = \omega_{FS} + i\del\delbar f^*\phi.
    \]
    Hence, by the Calabi-Yau theorem we have $f^*\phi$ and $\phi$ differs by a constant, which we denote by $c_f$.
    Note that $c_{f^{-1}} = - c_f$, and therefore
    $$\frac{1}{|\mathcal{S}|}\sum_{f\in \mathcal{S}} f^*\phi= \frac{1}{|\mathcal{S}|}\sum_{f\in \mathcal{S}} \phi + c_f
    = \phi + \frac{1}{|\mathcal{S}|}\sum_{f\in \mathcal{S}} c_f
    =  \phi.
    $$
    For an arbitrary element $h \in \mathcal{S}$, then
    \[
    h^* \phi =  \frac{1}{|\mathcal{S}|}\sum_{f\in \mathcal{S}} h^*(f^*\phi) 
    = \frac{1}{|\mathcal{S}|}\sum_{f\in \mathcal{S}} (f \circ h)^*\phi 
    = \frac{1}{|\mathcal{S}|}\sum_{\Tilde{f}\in h\cdot\mathcal{S}} \Tilde{f}^*\phi 
    = \phi,
    \]
    where the last equality is due to $\mathcal{S} \simeq h\cdot \mathcal{S}$.
   
\end{proof}

\section{Computing volume predictions}\label{ap:integration}
For each of the experiments, we check the global consistency of the prediction by computing the volume. This requires numeric integration, which can be performed using Monte-Carlo methods. We follow the procedure of the \cymetric package, which is described in \cite{Larfors:2022nep}. Thus, over a sample of points, the integral of a function $f$ is computed as the weighted sum
\begin{align}\label{intf1}
  \int_X\text{d}\vCY\; f=\int_X dA\,\frac{\text{d}\vCY}{dA}\; f\simeq\frac{1}{N}\sum_{i=1}^N w_i f(p_i)\; ,
\end{align}
where $dA$ is the Fubini--Study measure, and the weights of each point are defined as
\begin{align}
 w_i=\left.\frac{\text{d}\vCY}{dA}\right|_{p_i} =\left.\frac{\Omega \wedge \overline{\Omega}}{dA}\right|_{p_i}  \; .
\end{align}
These weights are computed by the point generator routines of the \cymetric package.

In order to check the global accuracy of the learned metric, we will instead compute
\begin{align}\label{intf2}
\int_X\text{d}\vCY\; f\simeq 
\frac{1}{N}\sum_{i=1}^N\tilde{w}_i\, 
{\rm det}( \gCY(p_i))\, f(p_i)\;,
\end{align}
using the auxiliary weights $\tilde{w}_i$, defined as $\tilde{w}_i = w_i/(\Omega \wedge \overline{\Omega})|_{p_i}$.

\section{Computing error measures}\label{ap:measures}
For each experiment, we check the accuracy of the prediction by computing the $\sigma$ and Ricci measure which are given by 
\begin{align}
    \sigma &= \frac{1}{\vCY} \int_X \left\lVert 1 - \frac{1}{\kappa} \frac{J \wedge J \wedge J}{\Omega \wedge \overline{\Omega}} \right\rVert_1 \,, 
\end{align}
where $\vCY = \int_X \d \vCY$ and $\kappa=(\int_X J \wedge J \wedge J)/\vCY$, as well as
\begin{equation}
    \mathcal{R} = \frac{(\int_X J \wedge J \wedge J)^{1/3}}{\vCY} \int_X \left\lVert \del \delbar \ln {\rm det}(\gCY) \right\rVert_1 \,.
\end{equation}
Following the default settings of the \cymetric package, the integrals are numerically evaluated as in Eq. (\ref{intf1}), except for the integral over the Ricci-scalar which is evaluated as in Eq. (\ref{intf2}).

\section{Computational performance}
Some of the experiments were run on an Intel\textsuperscript{\textcircled{\tiny R}}  Core\textsuperscript{TM} i7-10700K CPU with 8 cores. 64 GB of RAM were available during all experiments. Other experiments were run on clusters provided by NAISS, using either different CPU configurations or a GPU. 

The homogeneous layer combined with the permutation and root scaling layer trained on 90~000 points for $50$ epochs used a maximum of 11 GB of RAM and took approximately $45$ minutes to train. This RAM usage is higher than for the dense layer, and running experiments on a standard laptop with 16 GB of RAM would often crash due to lack of memory.
For comparison, the runtimes reported for the dense \cymetric models in Refs.~\cite{Larfors:2021pbb,Larfors:2022nep,Anderson:2023viv} were about one hour on a laptop.

\clearpage

\bibliographystyle{bibstyle}

\providecommand{\href}[2]{#2}

\end{document}